\documentclass[aps,prx,reprint,superscriptaddress,longbibliography,nofootinbib,floatfix]{revtex4-2}

\usepackage{amsmath,amssymb,amsfonts,amsthm,mathtools}
\usepackage{mathrsfs}
\usepackage{bm}
\usepackage{graphicx}
\usepackage{array}[=2016-10-06]
\usepackage{microtype}
\usepackage{hyperref}
\hypersetup{hidelinks}
\usepackage{xcolor}

\newenvironment{revision}{}{\par}
\usepackage{etoolbox}
\usepackage{tikz}
\usetikzlibrary{arrows.meta}

\graphicspath{{./}{./PRX_Quantum_numerical_codes/teleportation_experimental_model/}}

\DeclareMathOperator{\id}{id}
\DeclareMathOperator{\Tr}{Tr}
\DeclareMathOperator{\Lie}{Lie}

\newcommand{\cH}{\mathcal H}
\newcommand{\ket}[1]{\lvert #1\rangle}
\newcommand{\bra}[1]{\langle #1\rvert}
\newcommand{\ketbra}[1]{\lvert #1\rangle\!\langle #1\rvert}
\newcommand{\norm}[1]{\lVert #1\rVert}

\theoremstyle{plain}
\newtheorem{theorem}{Theorem}
\newtheorem{proposition}[theorem]{Proposition}
\newtheorem{lemma}[theorem]{Lemma}
\newtheorem{corollary}[theorem]{Corollary}
\theoremstyle{definition}
\newtheorem{definition}[theorem]{Definition}
\theoremstyle{remark}

\begin{document}

\title{Two-body control of noiseless subsystems with mixed \texorpdfstring{$SU(3)$}{SU(3)} representations lifts teleportation above the classical fidelity limit}

\author{Otavio A. D. Molitor}
\email[Corresponding author: ]{dmolitor.oa@protonmail.com}
\affiliation{International Centre for Theory of Quantum Technologies (ICTQT), University of Gda\'nsk, 80-308 Gda\'nsk, Poland}

\author{Sergii Strelchuk}
\email{sergii.strelchuk@cs.ox.ac.uk}
\affiliation{Department of Computer Science, University of Oxford, Parks Road, Oxford OX1 3QG, United Kingdom}

\author{Micha\l{} Studzi\'nski}
\email{michal.studzinski@ug.edu.pl}
\affiliation{International Centre for Theory of Quantum Technologies (ICTQT), University of Gda\'nsk, 80-308 Gda\'nsk, Poland}

\date{\today}

\begin{abstract}
Collective noise can leave degrees of freedom unaffected, allowing them
to store quantum information. Processing this information also requires
physically available logical gates. We consider $r$ fundamental qutrits
transforming under $U\in SU(3)$ and $s$ dual qutrits transforming under
its complex conjugate $\overline U$, with collective action
$R_{r,s}(U)=U^{\otimes r}\otimes\overline U^{\otimes s}$.
Using the image of the walled Brauer algebra and Lie-algebra arguments,
we show that adjacent exchanges within each qutrit type and one
singlet-projector interaction between the two types generate any
independently chosen family of special unitary gates on the noiseless subsystems by a finite pulse sequence, for every $r,s\geq1$.
This establishes semi-universality of this system. We illustrate the result with a numerically optimised sequence of 24 pulses for three fundamental qutrits and one dual qutrit. We then use the encoding and available
logical corrections for qutrit teleportation under correlated Weyl noise. Under explicit assumptions on channel composition and which trials are retained, and allowing for baseline implementation errors, we identify noise ranges in which encoded teleportation exceeds the
classical fidelity bound of $1/2$, while the unencoded teleportation channel is entanglement breaking. We analyse fidelity on accepted trials separately from the transmission costs of the additional physical qutrits.
\end{abstract}

\maketitle

\section{Introduction}
\label{sec:introduction}

Noise does not always lead to negative outcomes of computation. Collective noise, that is, noise that acts equally on all subsystems of a given quantum system can induce a separation into degrees of freedom that are and are not affected by it. The Hilbert space decomposes into
symmetry sectors, each with a gauge factor affected by the noise and a
multiplicity factor on which the noise acts as the identity. These
multiplicity factors are noiseless subsystems, and encoding quantum
information in them is a form of passive quantum error correction
\cite{ZanardiRasetti1997,LidarChuangWhaley1998,KnillLaflammeViola2000}.
In a decoherence-free subspace (DFS), every encoded state is unchanged
by the noise. A noiseless subsystem allows the gauge state to change
while the logical state remains unchanged.
When collective errors dominate, this form of passive error
correction can reduce the demands on active error correction. Its
usefulness depends, however, on being able to manipulate the encoded
information as well as preserve it. The combination of passive suppression and active correction has
been demonstrated in a different setting: stabilised cat qubits suppress
bit flips, while an outer repetition code corrects phase
errors~\cite{PuttermanEtAl2025}.

Noiseless encodings already have several uses in quantum communication.
Without a shared reference frame, an unknown collective rotation can make
a single transmitted qubit completely uninformative to the receiver.
Several qubits can nevertheless carry quantum information faithfully if
it is stored in a noiseless subsystem~\cite{BartlettRudolphSpekkens2003}.
DFS encodings have also been proposed for teleporting logical
states~\cite{WeiDengZhangFeng2007}. In the photonic experiment of
Prevedel et al., teleportation was performed through measurements on a
cluster state, with the encoding protecting the information against
collective phase noise~\cite{PrevedelEtAl2007}. Other photonic experiments
distributed entanglement through noisy optical fibres and improved the
distribution efficiency in the presence of loss
\cite{YamamotoEtAl2008,IkutaEtAl2011}. Passive encodings can thus protect
not only a transmitted state, but also its entanglement with a distant
system.

Using a noiseless encoding for quantum information processing also
requires physically available logical gates. Symmetry-preserving
interactions may generate only a restricted set of logical
operations~\cite{Marvian2022}. The aim is to find interactions that respect the collective symmetry and provide sufficiently general logical
control while acting on only a few physical systems at a time. For qubits, two-body exchange interactions achieve universal quantum computation with suitable encodings~\cite{BaconKempeLidarWhaley2000,DiVincenzoEtAl2000,
KempeBaconLidarWhaley2001,KempeBaconDiVincenzoWhaley2001}. Exchange-only logical gates have been demonstrated with encoded spin qubits in silicon~\cite{WeinsteinEtAl2023}, with recent processors
integrating cryogenic digital control~\cite{HRLQuantumTeam2026}.

We now consider qutrits, the simplest systems beyond qubits. Independent logical control means that a gate on one logical subsystem can be chosen without fixing the gates applied to the others. This is a stronger
requirement than performing arbitrary gates in a single chosen encoding, as in exchange-only constructions for encoded qudits~\cite{VanMeter2021}. In general, for qudits transforming under the same $U\in SU(d)$, symmetry-preserving interactions between pairs of subsystems do not provide independent logical control and can limit the logical gates available within a sector due to additional conservation laws~\cite{MarvianLiuHulse2021,Marin2007}. In the case of qutrits, three-body interactions are sufficient for independent logical control, although restrictions on relative phases between sectors remain~\cite{HulseLiuMarvian2024}. This property is called semi-universality: special-unitary gates can be chosen independently on the logical subsystems, without requiring arbitrary independent relative phases between sectors. Hulse \textit{et al.}~\cite{HulseLiuMarvian2024} give two-body constructions using six ancillas for
semi-universality on a fundamental-qutrit register and nine for full
invariant universality, with the ancillary state restored at the end~\cite[Corollary~3]{HulseLiuMarvian2024}. These counts are sufficient,
not asserted to be minimal.

We ask whether interactions between pairs of subsystems can provide independent logical control throughout a mixed register, without additional initialised ancillas, by combining qutrits with different transformation
rules under the same collective noise. We answer this question for systems in which some qutrits transform
under \(U\in SU(3)\), known as fundamental qutrits, and others, the dual qutrits, under its complex conjugate
\(\overline U\).
Writing \(V\cong\mathbb C^3\) for the
first representation and \(V^*\) for the second, a system with \(r\)
qutrits of the first type and \(s\) of the second has global Hilbert space
\(\cH_{r,s}=V^{\otimes r}\otimes(V^*)^{\otimes s}\), with mixed-tensor representation of the noise:
\begin{equation}
  R_{r,s}(U)=U^{\otimes r}\otimes\overline U^{\otimes s},
  \qquad U\in SU(3).
  \label{eq:collective-action-intro}
\end{equation}
The same \(U\) determines the transformation of every qutrit, so the noise remains collective even though the two types respond differently. The existence of two kinds of qutrits makes an additional interaction available. A pair
containing one qutrit of each type has a maximally entangled singlet that is unchanged by \(U\otimes\overline U\). An interaction proportional
to the projector onto this singlet respects the collective symmetry from Eq.~\eqref{eq:collective-action-intro}, as do exchange interactions between qutrits of the same type. Our main theorem shows that these interactions provide independent
logical control over every multiplicity sector for every \(r,s\geq1\). More precisely, any logical unitary with determinant one can be engineered from the three kinds of physical interactions independently on each
noiseless subsystem. Within a chosen sector, this
includes every logical gate up to an overall phase. This ensures semi-universality as defined in~\cite{HulseLiuMarvian2024} for the set of physically available gates. The result and its proof overview appear in
Sec.~\ref{sec:global-semiuniversality}. We illustrate it with four qutrits,
three fundamental and one dual, in Sec.~\ref{subsec:explicit-sector-control}.
There we numerically optimise a sequence of 24 pulses implementing a
cyclic shift on one logical qutrit while leaving the remaining logical
subsystems unchanged, to numerical precision. The theorem establishes
the existence of finite sequences for arbitrary independent special-unitary
targets, without supplying an efficient general synthesis algorithm.

% \revxxiv{Noise correlations can assist communication through other mechanisms.
% In Pirandola's model, two systems can retain their mutual entanglement after
% transmission through a correlated environment even though each local
% channel, considered separately, is entanglement
% breaking~\cite{Pirandola2013Reactivation}. Correlated environments can
% also produce memory effects that improve teleportation and, in suitable
% limits, permit perfect recovery~\cite{LaineBreuerPiilo2014}. Fidelity
% improvements have been studied for both qubits and qutrits subject to
% correlated errors~\cite{GuoEtAl2020,XuEtAl2022}. These mechanisms rely on
% different assumptions. They nevertheless show why the effect of noise on
% individual systems can underestimate the communication possibilities when
% correlations are present.}
% \revxxv{Correlation-assisted bosonic communication has also been analysed
% theoretically and demonstrated in a proof-of-principle
% experiment~\cite{PirandolaEtAl2021}, using a relay architecture rather
% than the noiseless encoding considered here.}

We apply our framework to noisy qutrit teleportation. Qutrit teleportation
has already been demonstrated with photonic systems
\cite{LuoEtAl2019,HuEtAl2020}. Before entanglement distribution, one member
of an entangled qutrit pair is encoded into a logical qutrit stored in
four physical qutrits, three fundamental and one dual. The collective
part of the noise leaves this logical qutrit, including its entanglement
with the other member of the pair, unchanged. The control theorem
supplies the logical corrections required by the Bell measurement.
These gates act within the chosen sector and do not cause leakage.
To compare encoded and unencoded teleportation, we use a common noise
model built from Weyl errors, the qutrit analogues of Pauli errors.
The model includes shared random transformations of the form in
Eq.~\eqref{eq:collective-action-intro}, together with Weyl labels chosen
independently for different physical qutrits. The conjugate Weyl operator
acts on the dual qutrit. Encoding removes the effect of shared
transformations, while independently chosen errors can disturb the
logical state. Preparation, encoding and recovery can introduce further
imperfections. The comparison asks when protection against collective
noise outweighs these additional costs.

We identify conditions under which the encoded protocol exceeds the
classical teleportation fidelity bound of \(1/2\) throughout a finite
interval of noise strengths for which unencoded teleportation is
entanglement breaking. This bound is the highest average fidelity obtainable by measuring an unknown, uniformly distributed pure qutrit
state and using the classical result to prepare an output state. An entanglement-breaking channel has the stronger property that it
destroys any entanglement between the transmitted qutrit and an external system. In the interval we identify, the encoded protocol
retains the ability to transmit entanglement after the unencoded
protocol has lost it. The advantage can persist when noise is only partly collective
and the encoded protocol has a lower fidelity in the absence of
transmission noise. We derive the required conditions on the noise correlations and on the performance of the additional operations. The comparison uses a model of the complete protocol and assumptions
about which trials are retained, specified in
Sec.~\ref{sec:baseline-calibration}. 

The paper is organised as follows. Section~\ref{sec:noiseless-control}
introduces noiseless subsystems, symmetry-preserving gates and the
available pair interactions. Section~\ref{sec:global-semiuniversality}
gives Theorem~\ref{thm:global-semiuniversality}, outlines its proof and
presents a logical qutrit encoded in four physical qutrits, together
with a numerical pulse sequence implementing a logical $X$ gate.
Appendices~\ref{app:mathematical-framework}--\ref{app:propagation-base}
give the mathematical framework and complete proof.
Section~\ref{sec:teleportation-model} compares encoded and unencoded
qutrit teleportation, including fidelity and transmission costs.
Appendix~\ref{app:teleportation-analysis} gives the channel derivations
and resource analysis. Section~\ref{sec:discussion} discusses the
results and future directions.

% The control theorem assumes tunable interactions and access to both
% qutrit types. Section~\ref{sec:heralded-cost} discusses a five-qutrit
% realisation using only ordinary qutrits and the additional requirements
% of state preparation, encoding, Bell measurement and recovery.
% We report fidelity on accepted trials separately from the probability
% that a trial produces an output, since using more physical qutrits
% increases sensitivity to loss. The resulting fidelity advantage
% establishes a regime in which encoded teleportation can transmit
% entanglement, while a comparison of communication rates at fixed
% physical resources requires these additional costs to be included.

\section{Noiseless qutrits and physically accessible interactions}
\label{sec:noiseless-control}

Consider a fixed symmetry sector with a gauge subsystem and a logical
subsystem. For a product state, collective noise acts as
\begin{equation}
 \tau_{\mathrm G}\otimes\rho_{\mathrm L}
 \longmapsto
 [\varrho(U)\tau_{\mathrm G}\varrho(U)^\dagger]
 \otimes\rho_{\mathrm L},
\end{equation}
where $\varrho(U)$ represents $U$ on the gauge subsystem. Tracing out
the gauge subsystem returns $\rho_{\mathrm L}$, even when the rotation
is unknown or randomly chosen. The logical reduced state remains
unchanged even when the gauge and logical subsystems are correlated,
because the noise acts only on the gauge factor. The same statement
holds when the logical subsystem is entangled with an external system.
The product form above is illustrative and does not imply a product
state of the physical qutrits. This separation of gauge and logical
information is part of the operator quantum error-correction
framework~\cite{KribsLaflammePoulin2005}.

Mixed Schur--Weyl duality gives the decomposition of the full Hilbert space:
\begin{equation}
 \cH_{r,s}
 \cong
 \bigoplus_{(p,q)\in\Lambda_{r,s}}
 U_{p,q}\otimes W_{r,s}^{p,q}.
 \label{eq:gauge-logical-decomposition}
\end{equation}
Here $U_{p,q}$ carries an irreducible representation of $SU(3)$,
specified by the nonnegative integers $(p,q)$, known as Dynkin labels.
The space $U_{p,q}$ is the gauge factor. The dimension of
$W_{r,s}^{p,q}$ counts how many copies of $U_{p,q}$ occur in
$\cH_{r,s}$. Collective transformations act identically on these
copies and as the identity on the multiplicity space $W_{r,s}^{p,q}$,
where logical information is encoded. A sector with
$\dim W_{r,s}^{p,q}=3$, for example, contains a logical qutrit.
The set $\Lambda_{r,s}$ contains the labels present in the decomposition.

Collective transformations and symmetry-preserving gates $X$ act on complementary factors:
\begin{align}
 R_{r,s}(U)&=\bigoplus_{(p,q)\in\Lambda_{r,s}}
 \varrho_{p,q}(U)\otimes I_{W_{r,s}^{p,q}},
 \label{eq:noiseless-action-master}\\
 X&=\bigoplus_{(p,q)\in\Lambda_{r,s}}
 I_{U_{p,q}}\otimes X_{p,q},
 \label{eq:invariant-logical-action-main}
\end{align}
where $\varrho_{p,q}(U)$ is the matrix representing $U$ on the
gauge subsystem. Equation~\eqref{eq:invariant-logical-action-main}
follows from Schur's lemma for every unitary $X$ satisfying
$[R_{r,s}(U),X]=0$ for all $U\in SU(3)$. Such a gate acts on the
logical factor and leaves the gauge state unchanged. We ask whether
interactions between pairs of physical qutrits can implement the
logical gates independently in different sectors.

For the full collective $SU(3)$ action, a decoherence-free subspace
occurs when the gauge representation is one-dimensional. This happens
only in the trivial sector $(p,q)=(0,0)$, when that sector is present.
In the other sectors, the logical subsystem is noiseless while the
gauge subsystem can evolve. Appendix~\ref{app:mathematical-framework}
develops the logical block description.
Appendix~\ref{app:branching-multiplicity} constructs the multiplicity
spaces and explains how they change when either type of qutrit is added.

We now describe the physically available interactions. For two fundamental qutrits we assign the exchange
Hamiltonian \(S_{ij}\), a dimensionless operator that swaps states of fundamental qutrits.
On the other hand, the dimensionless exchange Hamiltonian \(S_{\bar i\bar j}\) swaps states of dual qutrits (bars over indices indicate sites of dual qutrits). Formally:
\begin{equation}
    S_{ij}|a\rangle_i|b\rangle_j = |b\rangle_i|a\rangle_j,\qquad \forall a,b \in \{0,1,2\}.
    \label{eq:exchangeFundamental}
\end{equation}
\begin{equation}
    S_{\bar{i}\bar{j}}|\overline{a}\rangle_{\overline{i}}|\overline{b}\rangle_{\overline{j}} = |\overline{b}\rangle_{\overline{i}}|\overline{a}\rangle_{\overline{j}}, \qquad \forall a,b \in \{0,1,2\}.
    \label{eq:exchangeDual}
\end{equation}
Exchanges provide the invariant pair interaction within each qutrit
type. A pair containing one fundamental and one dual qutrit admits a
different invariant interaction. Such a pair has the normalised singlet:
\begin{equation}
\begin{aligned}
 \ket{\Omega_{i\bar j}}
 &=
 \frac{1}{\sqrt{3}}
 \sum_{a=0}^{2}|a\rangle_i|\bar a\rangle_{\bar j},\\
 P_{\Omega_{i\bar j}}
 &:=
 |\Omega_{i\bar j}\rangle\langle\Omega_{i\bar j}|,
 \qquad
 K_{i\bar j}:=3P_{\Omega_{i\bar j}}.
\end{aligned}
\label{eq:singlet}
\end{equation}
The vectors $\ket a$ and $\ket{\bar a}$ denote the corresponding basis
states in $V$ and $V^*$. The singlet is unchanged by the collective action,
$(U\otimes\overline U)\ket{\Omega_{i\bar j}}
=\ket{\Omega_{i\bar j}}$, so $K_{i\bar j}$ respects the same symmetry.
The pair space decomposes into a singlet and an octet,
$\mathbb C^3\otimes(\mathbb C^3)^*\cong\boldsymbol{1}\oplus\boldsymbol{8}$.
A pulse $e^{-itK_{i\bar j}}$ multiplies the singlet component by $e^{-3it}$
and leaves the octet component unchanged. It thus changes the relative
phase between these components without measurement or postselection.

Up to a real scale factor and an additive identity term, these exchanges
and singlet interactions exhaust the invariant pair Hamiltonians.
Their finite pulse products generate the physical control group.
Each pulse has the form $e^{-itH}$, where $H$ is an available interaction
and $t\in\mathbb R$ is a continuously adjustable parameter. The model
assumes tunable interactions. Their physical realisation, including
access to the conjugate transformation, depends on the implementation.

\section{Independent logical control}
\label{sec:global-semiuniversality}

The interactions in Sec.~\ref{sec:noiseless-control} implement any
special-unitary gate on a chosen logical subsystem, independently of
the gates on the other logical subsystems. This holds whenever both
qutrit types are present.

\begin{theorem}[Control using two-body interactions]
\label{thm:global-semiuniversality}
For every \(r,s\geq1\) and every independently chosen family
\(X_{p,q}\in SU(W_{r,s}^{p,q})\), with
\((p,q)\in\Lambda_{r,s}\), there is a finite sequence of $N$ exchange
and singlet pulses such that, in the decomposition of
Eq.~\eqref{eq:gauge-logical-decomposition},
\begin{equation}
 e^{-it_NH_N}\cdots e^{-it_1H_1}
 =\bigoplus_{(p,q)\in\Lambda_{r,s}}
 I_{U_{p,q}}\otimes X_{p,q}.
 \label{eq:global-independent-control}
\end{equation}
    Here each \(H_j\) is an allowed two-body interaction (see Eqs.~\eqref{eq:exchangeFundamental},\eqref{eq:exchangeDual} and \eqref{eq:singlet}) and \(t_j\) is a real parameter. The group \(SU(W_{r,s}^{p,q})\) contains all unitaries acting on \(W_{r,s}^{p,q}\) with determinant one.
\end{theorem}

Each pulse preserves the collective symmetry, so the sequence performs
the chosen logical gates without moving information out of its sector.
Within a chosen sector, this includes every logical unitary up to an
overall phase. The theorem does not establish arbitrary independent
control of relative phases between sectors. The independent choice of
special-unitary logical gates is called
semi-universality~\cite{HulseLiuMarvian2024}. When $r=s=1$, both logical
spaces are one-dimensional and the target is trivial.

Access to an interaction on every pair of qutrits is unnecessary.
Order the sites as $1,\ldots,r,\bar1,\ldots,\bar s$.
The $r+s-1$ adjacent interactions suffice: the exchanges $S_{a,a+1}$
for $1\leq a<r$, the singlet interaction $K_{r\bar1}$ between the
two types, and the exchanges $S_{\bar b,\overline{b+1}}$ for $1\leq b<s$.
Exchange pulses implement permutations within each type, up to a
scalar phase. Applying such a permutation before an interaction and
its inverse afterwards gives the corresponding interaction on any
allowed pair. The phases cancel in this conjugation.
Appendix~\ref{app:mathematical-framework} proves that the two sets of
interactions generate the same pulse group. This is a sufficient set,
without a claim of minimality. Implementing distant interactions
requires additional exchange pulses, and the theorem does not provide
an efficient general synthesis procedure.

\subsection{How the proof works}
\label{subsec:proof-assembly}

Generating the full associative commutant does not settle the control
problem. Pulse sequences are determined by the generated Lie algebra,
which can impose additional restrictions~\cite{LastresPollmannMoudgalya2026}.
The proof has three steps:
\begin{itemize}
\item[(I)] We extend control when a qutrit of either type is added.
The branching rules in Appendix~\ref{app:branching-multiplicity} express
each enlarged logical space as a direct sum of logical spaces from the
original system. Independent control on the original system supplies
gates on each contribution, with the identity on the other contributions
within that enlarged space. A sole contribution already supplies full
control of the enlarged logical space. If a contribution of dimension
at least three is a proper subspace, irreducibility ensures that some
interaction couples it to its complement. The argument isolates one new
direction and extends control to it. Repetition gives full special-unitary control within the
enlarged logical space (Lemmas~\ref{lem:propagation-irreducible-extension}
and \ref{lem:regular-sector-controllability}). Enlarged logical spaces
whose contributions all have dimension at most two require a separate
analysis (Lemma~\ref{lem:thin-sector-controllability}).

\item[(II)] Full control within each logical space does not by itself
imply that the gates in different sectors can be chosen independently.
For example, the group
\begin{equation*}
 \{L\oplus L:L\in SU(3)\}
\end{equation*}
gives every gate on either logical qutrit but cannot apply $L\oplus I$
for nontrivial $L$. Such links can occur only between logical spaces of
the same dimension, at least two
(Lemma~\ref{lem:propagation-semiuniversality-criterion}). In the induction
step, a gate inherited from the smaller system acts nonscalarly on one
enlarged logical space and as the identity on the other.
Lemma~\ref{lem:branching-separation} establishes this separation.
The four-qutrit starting case $(3,1)$ instead uses a short trace
comparison to establish independence of its two logical qutrits.

\item[(III)] Finally, Appendix~\ref{app:propagation-base} proves independent
logical control for the small systems $(r,s)=(1,1),(2,1),(2,2),(3,1)$.
Starting from $(3,1)$, Theorem~\ref{thm:mixed-propagation} extends the
result to every $r\geq3,s\geq1$.
\end{itemize}

Appendix~\ref{app:branching-multiplicity} derives the branching rules,
and Appendix~\ref{app:sectorwise-separation} proves the extension and
separation steps. Appendix~\ref{app:propagation-base} proves the initial
cases and carries out the induction. Exact Lie-algebra dimensions give
independent checks. Appendix~\ref{app:mathematical-framework} connects
the algebraic arguments to exact finite pulse sequences.

\subsection{A logical qutrit encoded in four physical qutrits}
\label{subsec:controllable-logical-qutrit}

We select the specific case of three fundamental qutrits and one dual qutrit, which have the global Hilbert space \(\cH_{3,1} = (\mathbb{C}^3)^{\otimes 3} \otimes (\mathbb{C}^3)^*\). The isotypic decomposition of this Hilbert space reads:
\begin{equation}
\begin{aligned}
 \cH_{3,1}\cong{}& U_{3,1}
 \oplus\bigl(U_{2,0}\otimes\mathbb C^3\bigr)\\
 &\oplus\bigl(U_{1,2}\otimes\mathbb C^2\bigr)
 \oplus\bigl(U_{0,1}\otimes\mathbb C^3\bigr).
\end{aligned}
 \label{eq:H31-logical-qutrit-decomposition}
\end{equation}
The gauge dimensions are \(24,6,15,3\) and the logical dimensions are \(1,3,2,3\), respectively. Indeed, the dimension of the full Hilbert space is 
\begin{equation*}
 3^4=81=24\cdot1+6\cdot3+15\cdot2+3\cdot3.
\end{equation*}
The logical spaces $W_{3,1}^{2,0}$ and $W_{3,1}^{0,1}$ both have
dimension three and each can encode a logical qutrit. They belong to
different sectors of the direct sum and are not simultaneously available
tensor factors. We choose \(W_{3,1}^{2,0}\) for
teleportation and for the pulse construction below. The whole sector has dimension \(18=6\cdot3\), with a six-dimensional gauge factor and a three-dimensional logical factor.

On the selected sector, the collective transformation $R_{3,1}(U)$ acts as
\begin{equation}
 R_{3,1}(U)\big|_{U_{2,0}\otimes W_{3,1}^{2,0}}=\varrho_{2,0}(U)\otimes I_3,
 \label{eq:logical-qutrit-complementary-actions}
\end{equation}
and the chosen logical gate $T_L$ acts as
\begin{equation}
    T_L\big|_{U_{2,0}\otimes W_{3,1}^{2,0}}=I_6\otimes L,\quad L\in SU(3).
\end{equation}
Theorem~\ref{thm:global-semiuniversality} allows the required logical gates to be performed within the chosen sector while leaving the other sectors unchanged.
For the considered system, the three two-body interactions
\(S_{12}\), \(S_{23}\), and \(K_{3\bar1}\) generate pulses implementing any $SU(3)$ gate on the logical subsystem.

To connect this result to teleportation, choose a logical basis
\(\{\ket{a}_{\mathrm L}\}_{a=0}^{2}\)
and write \(\omega=e^{2\pi i/3}\). The logical shift and phase operators are
\begin{equation}
\begin{aligned}
 X_{\mathrm L}\ket{j}_{\mathrm L}
 &=
 \ket{(j+1)\bmod3}_{\mathrm L},\\
 Z_{\mathrm L}\ket{j}_{\mathrm L}
 &=
 \omega^j\ket{j}_{\mathrm L},
 \qquad j\in\{0,1,2\}.
\label{eq:logicalGates}
\end{aligned}
\end{equation}
Both belong to $SU(3)$. From these operators one can construct logical Weyl operators as
\(D_g^{\mathrm L}:=X_{\mathrm L}^{a}Z_{\mathrm L}^{b}\),
with $g=(a,b)\in\mathbb Z_3^2$. These operators and their inverses
also belong to $SU(3)$. The control theorem thus supplies all logical
Weyl corrections required by qutrit teleportation, without causing
leakage from the chosen sector. 

Teleportation in this encoding requires control of the chosen logical
qutrit. Theorem~\ref{thm:global-semiuniversality} gives the stronger
property that its gates can be chosen independently of those on every
other logical subsystem. The following pulse sequence illustrates this
independence. State preparation, encoding, Bell measurement and recovery
are additional requirements of the teleportation protocol, discussed in
Sec.~\ref{sec:teleportation-model}.

\subsection{A pulse sequence for independent logical control}
\label{subsec:explicit-sector-control}

Both \((\mathbb{C}^3)^{\otimes4}\) and \((\mathbb{C}^3)^{\otimes3}\otimes(\mathbb{C}^3)^*\) have two multiplicity spaces
of dimension three. Their control properties differ. On four fundamental
qutrits, two-body invariant gates allow arbitrary special-unitary
operations within either space, but the two actions are linked by complex
conjugation, up to a fixed basis change and scalar phases~\cite[Sec.~IV~B, Eq.~(41)]{HulseLiuMarvian2024}. Hence, a cyclic
shift in one space cannot be paired with the identity in the other using only these gates on the four carriers. The gate set is thus not semi-universal. The mixed register has no such
restriction, since Theorem~\ref{thm:global-semiuniversality} guarantees independently chosen gates on all logical subsystems. We use this independence to apply a cyclic shift to \(W_{3,1}^{2,0}\), with the identity on \(W_{3,1}^{0,1}\) and on every remaining logical subsystem.

We numerically optimise a sequence of 24 pulses generated by \(S_{12}\), \(S_{23}\), and \(K_{3\bar1}\). Order the sectors $(p,q)$ as \((3,1),(2,0),(1,2),(0,1)\), and let
\(U_{\rm pulse}\) denote the direct sum of the induced unitaries on their logical spaces. The target logical gate is
\begin{equation}
 T_X
 =
 1\oplus X_3\oplus I_2\oplus I_3,
 \label{eq:selective-qutrit-target}
\end{equation}
where \(X_3\) is the matrix of the logical shift
\(X_{\mathrm L}\) from Eq.~\eqref{eq:logicalGates}, in the basis
specified in Appendix~\ref{subsec:explicit-pulses-31}.
The identity blocks ensure that the other logical subsystems remain unchanged, including the second logical qutrit \(W_{3,1}^{0,1}\).

We use the operator norm $\norm{\cdot}_{\rm op}$, defined for an operator $A$ as
\(\norm{A}_{\rm op}:=\max_{\norm{\psi}_2=1}\norm{A\psi}_2\),
where \(\norm{\cdot}_2\) is the Euclidean vector norm.
For \(A=U_{\rm pulse}-T_X\), this measures the largest state-vector error between the engineered pulse sequence and the target over all normalised inputs. High-precision evaluation of these blocks gives
\begin{equation}
 \norm{U_{\rm pulse}-T_X}_{\rm op}
 \approx 2.6\times10^{-15}.
 \label{eq:selective-pulse-error}
\end{equation}
This error is evaluated without adjusting an overall phase or separate phases
in different sectors. Appendix~\ref{subsec:explicit-pulses-31}
gives the logical bases, the reduced interaction matrices, and all pulse parameters. In this appendix we verify the reduced matrices directly from the physical interactions and reproduce the numerical pulse check.

This pulse sequence illustrates independent logical control within the ideal model of pair interactions on the two qutrit types. The quoted error of Eq.~\eqref{eq:selective-pulse-error} is obtained by numerical evaluation, rather than an analytic identity, and the sequence
has not been shown to minimise the number of pulses or tolerate control errors. The next section uses the same encoding to compare teleportation fidelity and transmission probability with and without our logical encoding.

\section{Recovering entanglement transmission}
\label{sec:teleportation-model}

We now compare encoded and unencoded teleportation under the same model of correlated noise. The unencoded protocol transmits one physical qutrit, while the encoded protocol stores one logical qutrit in the multiplicity subsystem
\(W_{3,1}^{2,0}\) of three fundamental and one dual qutrit. The encoded protocol thus transmits four physical qutrits. We identify conditions under which the encoding preserves the ability to transmit entanglement when the unencoded protocol has lost it.

\subsection{Teleportation protocol and recovery}

As the teleportation resource, we consider the maximally entangled state of two qutrits
\(\ket{\Phi_3}=3^{-1/2}\sum_{j=0}^{2}\,\ket j_A\ket j_B\). Alice retains subsystem \(A\), while the system holding \(B\) is sent to Bob through the noisy channel.
Only \(B\) is affected by this noise. In the unencoded protocol, \(B\) is a single physical qutrit. In the encoded protocol, \(B\) is first encoded into the multiplicity subsystem
\(W_{3,1}^{2,0}\) of three fundamental and one dual qutrit, with the gauge subsystem \(U_{2,0}\) prepared in a fixed state. These four physical qutrits then pass through the noisy channel.

After entanglement distribution, Alice performs a generalised Bell measurement on the unknown input and \(A\), then sends the measurement outcome to Bob~\cite{Bennett1993}. The unencoded protocol is the standard qutrit teleportation protocol. In the encoded protocol, Bob applies the corresponding logical Weyl correction, implemented by a pulse sequence supplied by Theorem~\ref{thm:global-semiuniversality}, and recovers a physical qutrit carrying the logical state. Applying the logical correction before ideal recovery has the same effect as applying the corresponding correction to the recovered qutrit. The logical Weyl correction can be performed while the information remains encoded.

On the selected sector, recovery discards the gauge subsystem and returns the logical state. Any component outside this sector, including leakage induced by noise, is replaced by the maximally mixed qutrit \(I_3/3\). This choice makes recovery trace preserving. The encoding and recovery maps, including the recovery action on states outside the selected sector, are given in Appendix~\ref{subsec:finite-weyl-reduction}.

We discard trials with a signalled physical loss or another detected failure. The remaining trials are accepted, and fidelity between input and output is averaged over these trials. The probability that a trial is accepted is reported separately. The assumptions needed to describe the accepted outputs by a quantum channel are given below.

\subsection{Shared and independently chosen Weyl errors}
\label{sec:weyl-noise-model}

Let \(\omega=e^{2\pi i/3}\), and define the physical qutrit
shift and phase operators, for \(j\in\{0,1,2\}\), by
\begin{equation}
 X\ket j=\ket{(j+1)\bmod3},
 \qquad
 Z\ket j=\omega^j\ket j.
\end{equation}
Define
\begin{equation}
    D_g:=X^aZ^b,\qquad g=(a,b)\in\mathbb Z_3^2.
\end{equation}
These are the nine Weyl operators, the qutrit analogues of Pauli
operators. Since \(\det X=\det Z=1\), each \(D_g\) belongs to
\(SU(3)\). Averaging uniformly over the Weyl operators completely
randomises a single qutrit:
\begin{equation*}
 \frac{1}{9}\sum_{g\in\mathbb Z_3^2}
 D_g\rho D_g^\dagger
 =
 \frac{I_3}{3}.
\end{equation*}

A shared error is a collective error that uses the same Weyl label
for all four qutrits, with the conjugate operator on the dual qutrit:
\begin{equation}
 R_{3,1}(D_g)
 =
 D_g^{\otimes3}\otimes\overline D_g.
\end{equation}
Let \(\mathcal T_{\rm col}\) average uniformly over the nine
choices of this shared label. As a departure from collective noise, we take \(\mathcal T_{\rm ind}\) to choose a label independently and uniformly for each of the four qutrits, using the conjugate operator on the dual qutrit. The noise channel is
\begin{equation}
 \mathcal N^{\rm Weyl}_{\varepsilon,\mu}
 =
 (1-\varepsilon)\id
 +\varepsilon\mu\,\mathcal T_{\rm col}
 +\varepsilon(1-\mu)\,\mathcal T_{\rm ind}.
 \label{eq:weyl-noise-channel}
\end{equation}
Here $\varepsilon,\mu\in[0,1]$, and $\id$ is the identity channel.
The parameter $\varepsilon$ is the probability of selecting
randomised noise. Within this branch, the labels are shared
with probability $\mu$ and chosen independently with probability
$1-\mu$. Since the uniform distribution includes the identity,
$\varepsilon$ is not the probability of a nonidentity error. Independence in \(\mathcal T_{\rm ind}\) refers to the Weyl labels. At \(\mu=0\), the channel is generally not a product channel. 

Table~\ref{tab:weyl-branch-actions} gives the action of each
branch. The encoded column includes recovery. Appendix~\ref{app:teleportation-analysis} gives a detailed derivation. A shared transformation leaves the encoded logical state unchanged because it has the collective form in Eq.~\eqref{eq:collective-action-intro}. Averaging over independently chosen labels completely randomises the four physical qutrits, after which recovery returns \(I_3/3\). For ideal encoding and recovery, only this branch changes the logical state. The unencoded qutrit is sensitive to both shared and independent labels. We write \(\Delta_3(\rho):=I_3/3\).

\begin{table}[tbp]
\caption{Ideal transmission channels for the three noise branches.
The encoded column includes recovery. Here \(\Delta_3\) is the completely depolarising qutrit channel.}
\label{tab:weyl-branch-actions}
\setlength{\tabcolsep}{3pt}
\begin{ruledtabular}
\begin{tabular}{l c c c}
Branch & Probability & Unencoded & Encoded\\ \hline
No randomisation & \(1-\varepsilon\) & \(\id\) & \(\id\)\\
Shared label & \(\varepsilon\mu\) & \(\Delta_3\) & \(\id\)\\
Independent labels & \(\varepsilon(1-\mu)\) & \(\Delta_3\) & \(\Delta_3\)
\end{tabular}
\end{ruledtabular}
\end{table}

Consider the depolarising channel of strength $\xi$:
\begin{equation}
 \mathcal D_\xi^{(3)}(\rho)
 =
 (1-\xi)\rho+\xi\frac{I_3}{3}.
\end{equation}
The weighted sum of the rows in Table~\ref{tab:weyl-branch-actions} gives
\(\mathcal D_\varepsilon^{(3)}\) for the unencoded protocol and
\(\mathcal D_{(1-\mu)\varepsilon}^{(3)}\) for the encoded protocol. These are the effective noise channels. Ideal teleportation gives these same channels because depolarising channels commute with the Weyl corrections~\cite{BowenBose2001}.

\subsection{Baseline performance and the channel model}
\label{sec:baseline-calibration}

The additional operations required by encoding can reduce fidelity before transmission noise is applied. For each protocol
\(\alpha\in\{{\rm u},{\rm enc}\}\), let
\(\mathcal E_\alpha^{(0)}\) be the quantum channel for the complete protocol on
accepted trials at \(\varepsilon=0\), and let
\(F_\alpha^{(0)}\) be its average fidelity. This baseline includes preparation, Bell measurement and correction for both protocols. For the encoded protocol, it also includes gauge preparation, encoding and recovery.

At fixed noise parameters, we assume that each protocol has the same positive probability of retaining a trial for every combination of input, noise branch and Weyl labels. This probability can differ between protocols and vary with the noise parameters. Dividing the map for accepted outcomes by this common probability gives a completely positive, trace-preserving (CPTP) channel, while leaving the sampled input distribution and the noise weights in Eq.~\eqref{eq:weyl-noise-channel} unchanged.

We separately assume that the complete channels on accepted trials have the following form:
\begin{equation}
\begin{aligned}
 \Lambda_{\rm u}
 &=\mathcal E_{\rm u}^{(0)}\circ\mathcal D_\varepsilon^{(3)},\\
 \Lambda_{\rm enc}
 &=\mathcal E_{\rm enc}^{(0)}
   \circ\mathcal D_{(1-\mu)\varepsilon}^{(3)},
\end{aligned}
 \label{eq:teleportation-channels-with-baseline}
\end{equation}
where each baseline channel \(\mathcal E_\alpha^{(0)}\) is fixed as the noise parameters vary and is independent of the noise branch and Weyl labels. This assumption applies to the complete protocol, including its Bell measurement and corrections. It does not follow from the retention condition or the baseline fidelity. Appendix~\ref{subsec:experimental-calibration-protocol} gives two channels with the same fidelity curve and different entanglement properties.

Let \(F_\alpha\) denote the average fidelity over all accepted trials, with pure qutrit inputs sampled uniformly with respect to Haar measure. When the depolarising part replaces the input by \(I_3/3\), the baseline channel produces a fixed output state. Its average overlap with the uniformly sampled inputs is \(1/3\), irrespective of the output state. Since fidelity with a pure input is linear in the output state, Eq.~\eqref{eq:teleportation-channels-with-baseline} gives
\begin{equation}
\begin{aligned}
 F_{\rm u}-\tfrac13
 &=(1-\varepsilon)(F_{\rm u}^{(0)}-\tfrac13),\\
 F_{\rm enc}-\tfrac13
 &=[1-(1-\mu)\varepsilon](F_{\rm enc}^{(0)}-\tfrac13).
\end{aligned}
 \label{eq:weyl-accepted-fidelities}
\end{equation}
In Appendix~\ref{subsec:weyl-threshold-derivation} we derive these
formulas. The baseline channels need not be unital or isotropic, and the complete channels need not be depolarising. Appendix~\ref{subsec:branch-resolved-data} gives the average fidelity when the probability of retaining a trial depends on the input, noise branch or Weyl labels. Appendix~\ref{subsec:experimental-calibration-protocol} explains how to test both the trial selection and the channel model.

\subsection{Entanglement transmission comparison}
\label{sec:entanglement-comparison}

An entanglement-breaking (EB) channel destroys entanglement between its input and every reference system~\cite{HorodeckiShorRuskai2003}.
For uniformly distributed
pure qutrit inputs, an EB channel has average fidelity at most
\(1/2\), exactly the classical fidelity bound obtained by measuring one copy
of the input and preparing an output from the measurement result \cite{Horodecki1999}. Exceeding this bound shows that a
channel can preserve entanglement. A fidelity at or below \(1/2\), however, does not establish that the channel is EB.

The qutrit
depolarising channel becomes entanglement breaking at
\begin{equation}
 \varepsilon_{\rm dep}^{\rm EB}=\frac34.
 \label{eq:depolarising-entanglement-breaking-boundary}
\end{equation}
A subsequent CPTP channel cannot restore entanglement, so
Eq.~\eqref{eq:teleportation-channels-with-baseline} guarantees
that \(\Lambda_{\rm u}\) is entanglement breaking whenever
\(\varepsilon\geq3/4\). Baseline errors may make the channel entanglement breaking at smaller values of \(\varepsilon\).
Appendix~\ref{subsec:depolarising-entanglement-breaking} derives this depolarising boundary.

It suffices to check the encoded fidelity at \(\varepsilon=3/4\), where the unencoded channel is already entanglement breaking. If the encoded fidelity is above \(1/2\), continuity gives a nonempty interval beyond this point in which only the encoded channel can preserve entanglement. This fidelity criterion applies while the encoded fidelity exceeds \(1/2\). It does not exclude an entanglement advantage at lower fidelities.
This sufficient comparison does not depend on
\(F_{\rm u}^{(0)}\), since any subsequent baseline channel leaves the unencoded channel entanglement breaking. This criterion establishes the comparison while the encoded
fidelity remains above $1/2$.

\begin{proposition}[Recovering entanglement transmission]
\label{prop:encoded-teleportation-entanglement}
Assume the noise model and the conditions on retained trials and
channel composition in Secs.~\ref{sec:weyl-noise-model}
and \ref{sec:baseline-calibration}. Let \(0\leq\mu<1\). If
\begin{equation}
 F_{\rm enc}^{(0)}>\frac{1+\mu}{1+3\mu},
 \label{eq:encoded-entanglement-condition}
\end{equation}
then the unencoded teleportation channel is entanglement breaking,
while the encoded channel has average fidelity above \(1/2\)
and is not entanglement breaking, for every
\begin{equation}
 \varepsilon\in
 [3/4,\varepsilon_{\rm enc}^{\rm MP})\cap[0,1],
 \label{eq:encoded-entanglement-interval}
\end{equation}
where the superscript \({\rm MP}\) refers to the measure-and-prepare
fidelity bound and
\begin{equation}
 \varepsilon_{\rm enc}^{\rm MP}
 =\frac{F_{\rm enc}^{(0)}-\tfrac12}
 {(1-\mu)(F_{\rm enc}^{(0)}-\tfrac13)}.
 \label{eq:encoded-threshold-master}
\end{equation}
\end{proposition}

\begin{proof}
The unencoded channel is entanglement breaking for
\(\varepsilon\geq3/4\), as above. For the encoded channel,
Eq.~\eqref{eq:weyl-accepted-fidelities} gives
\begin{equation*}
 F_{\rm enc}\big|_{\varepsilon=3/4}-\frac12
 =\frac{(1+3\mu)F_{\rm enc}^{(0)}-(1+\mu)}{4}.
\end{equation*}
This quantity is positive under
condition~\eqref{eq:encoded-entanglement-condition}, which gives
\(\varepsilon_{\rm enc}^{\rm MP}>3/4\).
The encoded fidelity decreases linearly with \(\varepsilon\)
and remains above \(1/2\) until
\(\varepsilon=\varepsilon_{\rm enc}^{\rm MP}\). An entanglement-breaking
qutrit channel cannot exceed this bound.
\end{proof}

For \(\mu=1\), collective noise leaves the logical qutrit
unchanged, so \(F_{\rm enc}=F_{\rm enc}^{(0)}\) for every
\(\varepsilon\). If \(F_{\rm enc}^{(0)}>1/2\), the same
entanglement comparison holds throughout \([3/4,1]\).

As an example, take
\begin{equation*}
 F_{\rm u}^{(0)}=\frac{29}{30},\qquad
 F_{\rm enc}^{(0)}=\frac{13}{15},\qquad
 \mu=0.20.
\end{equation*}
The unencoded channel is EB for
\(\varepsilon\geq3/4\), whereas the fidelity of the encoded protocol remains
above \(1/2\) for \(\varepsilon<55/64\simeq0.8594\).
The encoded protocol can thus transmit entanglement for the noise range
\(\varepsilon\in[3/4,55/64)\), where the unencoded channel is already entanglement breaking and destroys entanglement between the input and any reference system.
Figure~\ref{fig:experimental-teleportation} shows this interval. The parameters are chosen to illustrate the result and are not measurements from a device.

\begin{figure}[!tbp]
\centering
\includegraphics[width=\columnwidth]
{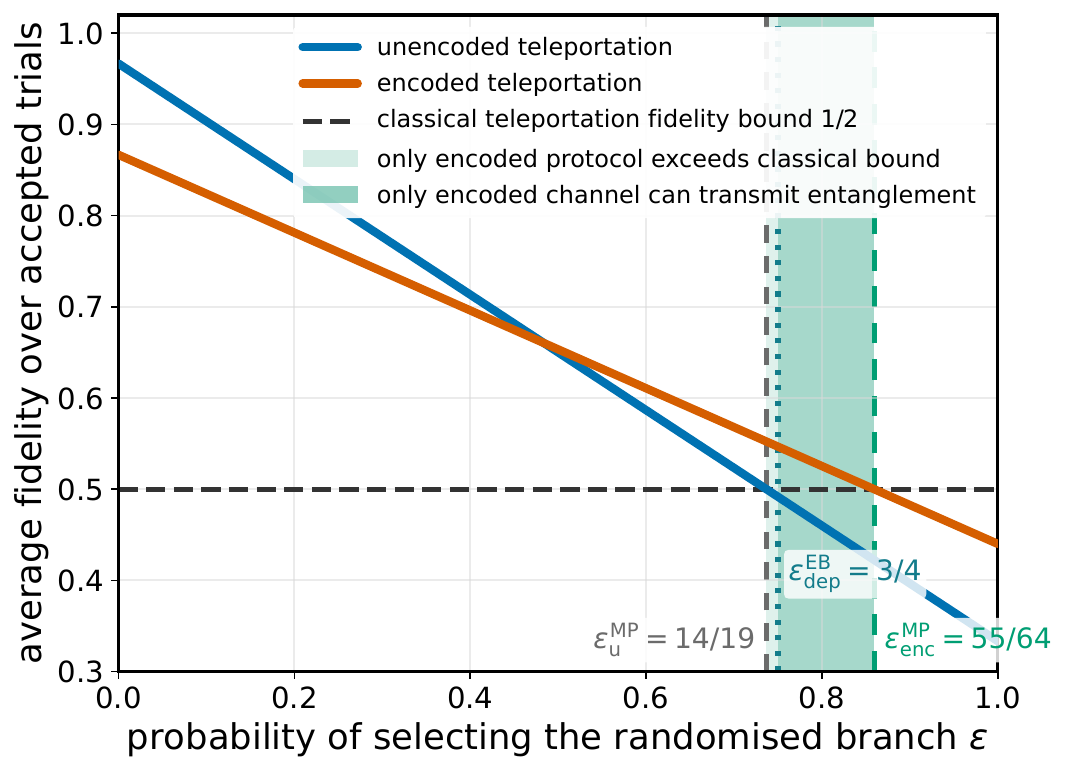}
\caption{Recovery of entanglement transmission through encoding.
In the darker interval \([3/4,55/64)\), the unencoded
teleportation channel is entanglement breaking, while the encoded
channel exceeds the classical fidelity bound and can preserve
entanglement. The curves use \(F_{\rm u}^{(0)}=29/30\),
\(F_{\rm enc}^{(0)}=13/15\), and \(\mu=0.20\).
The lighter shading begins at \(14/19\), where the unencoded
fidelity reaches \(1/2\).
Fidelities are averaged over accepted trials under the assumptions
of Sec.~\ref{sec:baseline-calibration}.}
\label{fig:experimental-teleportation}
\end{figure}

\subsection{Fidelity crossover and the classical bound}

The entanglement comparison in the previous subsection answers a different question from when encoding first improves fidelity. Subtracting the two curves gives
\begin{equation}
 F_{\rm enc}-F_{\rm u}
 ={}\varepsilon\mu(F_{\rm enc}^{(0)}-\tfrac13)
 -(1-\varepsilon)(F_{\rm u}^{(0)}-F_{\rm enc}^{(0)}).
 \label{eq:fidelity-gain-minus-cost}
\end{equation}
When \(1/3<F_{\rm enc}^{(0)}<F_{\rm u}^{(0)}\), the first
term is the gain from protection against collective noise. The subtracted term is the extra
baseline fidelity cost of encoding, weighted by the probability of the branch without randomisation. Encoding improves fidelity when this gain exceeds the baseline cost. For the example above, \(F_{\rm enc}-F_{\rm u}=(31\varepsilon-15)/150\), so the curves cross at
\(\varepsilon_\times=15/31\simeq0.4839\). Only the encoded
protocol has fidelity above \(1/2\) for \(\varepsilon\in[14/19,55/64)\).
Table~\ref{tab:reference-teleportation-thresholds} lists these four
noise strengths and their meanings.

An improvement in fidelity need not take either protocol above
the classical fidelity bound. Appendix~\ref{subsec:weyl-threshold-derivation}
gives the crossover formulas and the conditions for an interval
in which only the encoded fidelity exceeds the bound.

\subsection{Transmission probability and resources}
\label{sec:heralded-cost}

The encoding uses four physical qutrits when both fundamental and dual transformations are directly available. If every
physical qutrit is fundamental and transforms under \(U\), the conjugate action
can be represented within the antisymmetric subspace of two qutrits, using the isomorphism \(V^*\cong\wedge^2V\). This
implementation of \(\cH_{3,1}\) contains five physical qutrits instead of four, and the singlet interaction acts on three of them, as shown in
Eq.~\eqref{eq:encoded-singlet-interaction}. To reproduce the noise model,
the two qutrits in the antisymmetric pair must share a Weyl label
even in the branch with independently chosen labels. Appendix~\ref{subsec:experimental-calibration-protocol} gives this construction.

If each transmitted physical qutrit arrives independently with
probability \(\eta\), the probabilities that the complete system arrives
for the unencoded protocol and the encoded implementations with four
and five qutrits are, respectively,
\begin{equation}
 P_{\rm tr}^{\rm u}=\eta,\qquad
 P_{\rm tr}^{\rm enc,mixed}=\eta^4,\qquad
 P_{\rm tr}^{\rm enc,V}=\eta^5.
 \label{eq:link-arrival-probabilities}
\end{equation}
Figure~\ref{fig:teleportation-transmission} shows the cost of using additional qutrits. Related transmission costs and
ways of reducing them have been studied in photonic DFS
experiments~\cite{YamamotoEtAl2008,IkutaEtAl2011}.

\begin{figure}[!tbp]
\centering
\includegraphics[width=\columnwidth]
{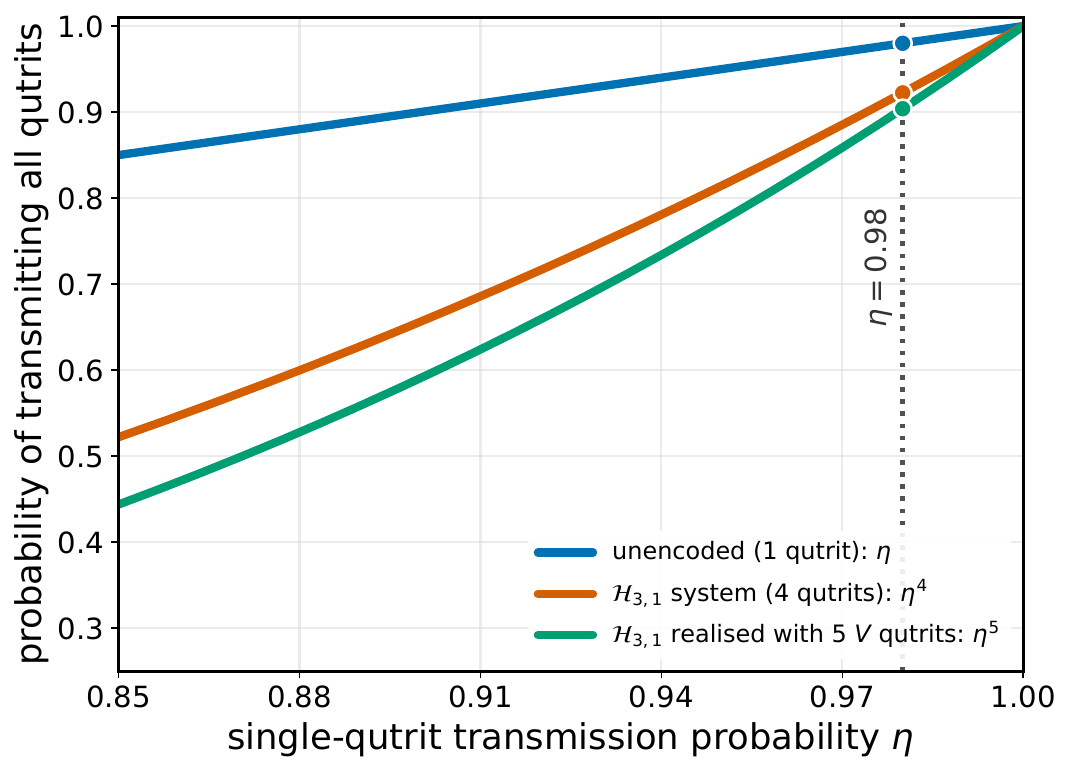}
\caption{Probability that every required qutrit arrives
when each transmitted qutrit arrives independently with probability
\(\eta\). The unencoded protocol uses one physical qutrit.
The encoded protocol uses four with direct access to both
transformation rules, or five when an antisymmetric pair
represents a dual qutrit. At \(\eta=0.98\), the
probabilities are \(0.98\), \(0.922368\), and
\(0.903921\), respectively. Preparation, detection, and processing failures are excluded. These transmission
probabilities are separate from the fidelities in Fig.~\ref{fig:experimental-teleportation}.}
\label{fig:teleportation-transmission}
\end{figure}

When loss is independent of the input, noise branch, Weyl labels and processing errors, discarding trials with a detected loss leaves the fidelity formulas unchanged. Other detected failures further reduce the fraction of accepted trials, so the transmission probabilities in Eq.~\eqref{eq:link-arrival-probabilities} are upper bounds on the total acceptance probabilities. These probabilities are measured separately from fidelity on accepted trials.
Probabilistic Bell measurements in photonic qutrit teleportation already require this distinction in the absence of transmission
noise~\cite{LuoEtAl2019,HuEtAl2020}.

The encoding scheme requires no additional shared entanglement. In both protocols, the Bell measurement has nine possible outcomes, communicated using two classical trits. Appendix~\ref{app:resource-accounting}
gives the full resource comparison. The fidelity advantage does not by itself establish an advantage in communication rate or channel capacity at fixed resources. Such a comparison
would also need to include repeated unencoded transmission with entanglement distillation~\cite{BennettPurification1996}
or an active error-correcting code.

\subsection{Relation to physical noise}
\label{sec:relation-physical-noise}

The Weyl noise model of Eq.~\eqref{eq:weyl-noise-channel} separates the effects of collective and independently chosen errors while keeping the teleportation channels exactly
calculable. Experiments on qutrit control and gate characterisation provide
methods for applying and characterising randomised gates \cite{MorvanEtAl2021,GossEtAl2022,GossEtAl2024}.
Experiments with decoherence-free encodings have likewise
examined shared and separately chosen transformations
\cite{AltepeterEtAl2004,YamamotoEtAl2008,IkutaEtAl2011}.
Universal qudit control has also been demonstrated with trapped
ions~\cite{RingbauerEtAl2022}. Qutrit and ququart protection has been
achieved through Heisenberg--Weyl dynamical
decoupling~\cite{TripathiEtAl2025} and, in bosonic encodings, through
error correction beyond break-even~\cite{BrockEtAl2025}. These experiments
provide control and protection methods, not a verification of the uniform Weyl mixture assumed here.

The noiseless encoding also protects against collective noise beyond the uniform Weyl distribution. Every collective transformation in
Eq.~\eqref{eq:collective-action-intro} acts only on the gauge
subsystem, irrespective of how \(U\) is sampled.
Spatial noise correlations have been measured in silicon spin qubits~\cite{YonedaEtAl2023}, although correlations alone do not imply
identical collective rotations.
Shared control fields or a common optical transformation can produce correlated rotations, although a particular implementation
may generate only a subset of \(SU(3)\). The corresponding actions \(U\) and \(\overline U\), and the singlet interaction
needed for logical gates, must still be realised physically.

Noiseless-subsystem protection has also been experimentally engineered on superconducting processors by combining dynamical decoupling and
error-detection postselection~\cite{QuirozEtAl2024}. The resulting protection relies on the applied control and postselection,
rather than on native collective noise alone.

The quantitative thresholds in this section rely on the Weyl mixture and the assumptions in Sec.~\ref{sec:baseline-calibration}.
They do not describe arbitrary device noise, such as local dephasing, energy relaxation or unequal rotations. Appendix~\ref{subsec:collective-extensions} discusses other collective processes, including continuous Markovian dynamics. It distinguishes collective processes that leave the logical state unchanged from comparisons based on measured decay rates, and bounds the fidelity loss from additional logical noise.

\section{Discussion and conclusions}
\label{sec:discussion}

\begin{revision}
Theorem~\ref{thm:global-semiuniversality} establishes independent logical
control for systems containing qutrits that transform under \(U\)
and others that transform under \(\overline U\), for all \(r,s\geq1\). Exchanges between qutrits
of the same type and the interaction through the invariant singlet of a
mixed pair suffice to choose a logical unitary with determinant one on every noiseless subsystem independently. All these interactions preserve the collective symmetry, so logical gates can be performed while the
information remains encoded in the logical subsystem. Within one sector, this includes every logical gate up to an overall phase. Arbitrary relative phases between sectors are not part of the result. The theorem thus extends the use of these noiseless subsystems from storage to logical computation.

Teleportation illustrates a communication application of the encoding and logical control.
Proposition~\ref{prop:encoded-teleportation-entanglement} identifies a
range of noise strengths for which the unencoded teleportation channel
is entanglement breaking, while the encoded protocol exceeds the
classical teleportation fidelity bound of \(1/2\). The encoded protocol
then remains capable of transmitting entanglement after the unencoded
protocol has lost this ability. The advantage can persist when the
noise is only partly collective and the additional encoding and recovery
operations lower the initial fidelity. This comparison follows under the assumptions about channel composition and retained trials in
Sec.~\ref{sec:baseline-calibration}. Encoding takes place before exposure
to the noise, allowing the logical information to avoid its collective
component.

The benefits of noiseless encodings for communication are already
established, including transmission without a shared reference frame
and teleportation in the presence of collective noise
\cite{BartlettRudolphSpekkens2003,PrevedelEtAl2007}.
Earlier work also shows that correlations in an environment can preserve
entanglement and improve teleportation
\cite{Pirandola2013Reactivation,LaineBreuerPiilo2014}.
Our contribution combines a general theorem on independent logical
control with an explicit comparison of encoded and unencoded qutrit
teleportation. The comparison relates the required noise correlations
to the performance of the complete protocol, including imperfections
in its additional operations.

A physical implementation requires preparation, encoding, Bell
measurement and recovery, together with access to the assumed
transformations and interactions. The teleportation encoding uses four
physical qutrits when the conjugate transformation \(\overline U\) is directly available.
Representing it within the antisymmetric subspace of two qutrits that both transform under \(U\), using \(V^*\cong\wedge^2V\), gives five physical qutrits and an interaction acting on three
of them. These requirements, and the resulting sensitivity to loss,
are discussed in Sec.~\ref{sec:heralded-cost}. Fidelity on accepted
trials describes the quality of the retained outputs. The probability
of retaining a trial determines how often an output is obtained.
Both are needed to assess an implementation.

Further work should develop preparation and recovery circuits for
specific devices, find shorter logical pulse sequences and assess their robustness. It is also natural to ask whether symmetry-based dynamical decoupling~\cite{ReadSerranoMartin2026} can suppress departures from
collective noise while preserving the logical gates available here.
Further experimental work should test the channel factorisation, check whether
discarding failed trials biases the input or noise distribution, and
determine how departures from collective noise change the entanglement advantage. Comparing communication rates or capacities at fixed physical resources requires
a separate analysis of these costs and alternative protocols. Such
work would establish when useful logical control and passive suppression of collective errors can be combined in a practical communication device.
\end{revision}

\begin{acknowledgments}
O.A.D.M. and M.S. acknowledge support from the IRA Programme, project
no.~FENG.02.01-IP.05-0006/23, financed by the FENG programme 2021--2027,
Priority FENG.02, Measure FENG.02.01, with the support of the FNP.  S.S.
acknowledges support from the Royal Society University Research Fellowship. The authors formulated the original theorem and its proof. They were helped by the generative AI (ChatGPT and Claude) to spot and rectify holes in the proof, to conduct the computation, code development, and writing of the manuscript. The authors bear complete responsibility for the findings and the manuscript.
\end{acknowledgments}
\appendix

Appendices~\ref{app:mathematical-framework}--\ref{app:propagation-base}
derive the branching rules and prove Theorem~\ref{thm:global-semiuniversality}
using analytic proofs for small systems followed by induction. Exact rank calculations and
the pulse example provide independent checks.
Appendix~\ref{app:teleportation-analysis} derives the effective teleportation
channels and fidelity thresholds, then explains how to average over
retained trials, test the noise model, compare physical resources,
and account for other noise processes.
\section{From invariant interactions to independent logical gates}
\label{app:mathematical-framework}

Three algebraic descriptions serve different purposes in the proof of Theorem~\ref{thm:global-semiuniversality}.
Sums and products of the two-body interactions show that no proper nonzero
subspace within a logical block is preserved by every interaction.
Their Lie algebra determines which gates their pulses can generate.
Its commutators are traceless on each logical block and test whether
determinant-one gates can be chosen independently. We also show why the Lie-algebra conclusion
gives exact finite pulse sequences. A dimension test supplies an independent check of the analytic proof.

\subsection{Logical spaces and invariant operators}
\label{sec:noiseless-subsystems-mixed-schur-weyl}

Fix \(r,s\geq0\). Within this appendix, write
\(\mathcal H=\mathcal H_{r,s}\), \(R=R_{r,s}\), and
\(\Lambda=\Lambda_{r,s}\). For a sector \(\lambda=(p,q)\),
write \(U_\lambda=U_{p,q}\), \(W_\lambda=W_{r,s}^{p,q}\),
and \(m_\lambda=\dim W_\lambda\). Choose a unitary isotypic decomposition
as in Eq.~\eqref{eq:gauge-logical-decomposition}:
\begin{equation}
 \mathcal H\cong\bigoplus_{\lambda\in\Lambda}
 U_\lambda\otimes W_\lambda.
 \label{eq:mixed-schur-weyl-decomposition}
\end{equation}
The irreducible representation \(U_\lambda\) occurs \(m_\lambda\) times, since the this is the dimension of each multiplicity subsystem $W_\lambda$.
Collective noise acts identically on these copies and as the identity
on the multiplicity space \(W_\lambda\), where logical information
is encoded.
The formal construction using intertwining maps is given in
Appendix~\ref{app:branching-multiplicity}.

Schur's lemma gives complementary actions for the noise and for any
invariant operator \(A\):
\begin{equation*}
 R(U)=\bigoplus_\lambda\varrho_\lambda(U)\otimes I,
 \qquad
 A=\bigoplus_\lambda I\otimes A_\lambda.
\end{equation*}
The matrices \(\varrho_\lambda(U)\) represent \(U\) in each sector
\(\lambda\), while the matrices \(A_\lambda\) are the logical blocks
of \(A\). When discussing logical control, we omit the identity
operators on the gauge spaces and identify \(A\) with
\(\bigoplus_\lambda A_\lambda\).
This identification preserves products, adjoints and commutators.
It also preserves linear independence, so physical matrices and
their logical blocks give the same Lie-algebra dimensions.
For an invariant group \(T\), write
\(T_\lambda=\{X_\lambda:X\in T\}\) for its action on one
logical space. We also write \([T]_\lambda\) when the group
name has indices or primes. The same block convention applies to an invariant
algebra or a set of operators. Later appendices distinguish original
and enlarged spaces explicitly and retain size indices in recursions.

% An invariant unitary has an arbitrary block in \(U(W_\lambda)\)
% on each logical space. Our target restricts each block to
% \(SU(W_\lambda)\), so it does not require arbitrary relative
% phases between sectors.

The dual action used in the main text has matrix
\((U^{-1})^{\mathsf T}=\overline U\).
The singlet \(\ket\Omega=3^{-1/2}\sum_a\ket a\ket{\bar a}\)
is invariant under \(U\otimes\overline U\), as follows from
\(UU^\dagger=I\). A bar on a basis or site label identifies
dual qutrits. Exchanges within each type and the
singlet interaction between the two types thus commute with
the collective action.

We collect these dimensionless two-body interactions in
\begin{equation*}
\begin{aligned}
 \mathcal C_{r,s}
 ={}&
 \{S_{ij}:1\leq i<j\leq r\}\\
 &\cup
 \{S_{\bar i\bar j}:1\leq i<j\leq s\}\\
 &\cup
 \{K_{i\bar j}:1\leq i\leq r,\ 1\leq j\leq s\}.
\end{aligned}
\end{equation*}

Mixed Schur--Weyl duality gives a stronger statement: the unital
complex associative algebra generated by these interactions is
the full commutant
\cite{GrinkoOzols2024,Nguyen2023MixedSchur,GrinkoBurchardtOzols2023}:
\begin{equation}
 \operatorname{Alg}_{\mathbb C}(\mathcal C_{r,s})
 =\operatorname{End}_{SU(3)}(\mathcal H)
 \cong\bigoplus_\lambda\operatorname{End}(W_\lambda).
 \label{eq:interaction-commutant}
\end{equation}
Here \(\operatorname{Alg}_{\mathbb C}\) allows complex linear
combinations and ordinary products, including the identity.
This operator algebra is the image of the walled Brauer algebra (WBA)~\cite{GrinkoThesis}.
Distinct elements of the abstract algebra can represent the same
physical operator. The representation map need not be injective.
The usual statement for \(U(3)\) also applies to \(SU(3)\):
every \(g\in U(3)\) has the form
\(g=zU\), with \(U\in SU(3)\) and \(|z|=1\), and
\(R(g)=z^{r-s}R(U)\). The scalar factor leaves the commutant
unchanged, so
\begin{equation*}
    \operatorname{End}_{SU(3)}(\mathcal H) = \operatorname{End}_{U(3)}(\mathcal H).
\end{equation*}

We use Eq.~\eqref{eq:interaction-commutant} to establish
irreducibility within each logical block: no proper nonzero
subspace is preserved by every interaction. Otherwise all their
sums and products would preserve it, which the full matrix algebra
does not. This guarantees a coupling out of any proper subspace
on which gates are already available. Lie-algebra arguments are
then needed to turn that coupling into arbitrary logical gates.

\subsection{Pulse sequences and the generated Lie algebra}
\label{sec:local-mixed-exchange-control-model}

Assume now that \(r,s\geq1\). The available finite pulse
sequences form the group
\begin{equation}
 G_{r,s}
 :=
 \left\langle
 e^{-itH}:H\in\mathcal C_{r,s},\ t\in\mathbb R
 \right\rangle
 \subseteq U(\cH_{r,s}),
 \label{eq:physical-control-group}
\end{equation}
where angle brackets denote finite products. We write \(G=G_{r,s}\)
when $r,s$ are fixed.

\begin{definition}[Independent logical control]
\label{def:global-mixed-semiuniversality}
The system admits independent logical control if every choice of
logical unitaries \(X_{p,q}\in SU(W_{r,s}^{p,q})\) can be implemented
simultaneously by a finite sequence of available pulses.
Equivalently, \(G_{r,s}\) contains the group
\begin{equation}
 \mathcal U_{\mathrm{log}}(r,s)
 :=\Bigg\{
 \bigoplus_{(p,q)\in\Lambda_{r,s}}
 I_{U_{p,q}}\otimes X_{p,q}:
 X_{p,q}\in SU(W_{r,s}^{p,q})
 \Bigg\}.
 \label{eq:logical-control-target}
\end{equation}
In logical blocks, the target is \(\prod_\lambda SU(W_\lambda)\).
It includes a gate on any one logical space with identity on all
the others and on every gauge space. On any one logical space,
every unitary is available up to a scalar phase.
\end{definition}

Let \(H_1,\ldots,H_N\), with \(N=r+s-1\), be the adjacent
two-body interactions: exchanges within each qutrit type and
the singlet interaction \(K_{r\bar1}\) between the two types.
Their pulses generate the same group as pulses from interactions
on all allowed pairs. Indeed, \(S^2=I\) gives
\(e^{-i\pi S/2}=-iS\), so adjacent exchange pulses implement
any permutation within a type, up to a scalar phase.
Conjugating an adjacent exchange or \(K_{r\bar1}\) by such
permutations gives the corresponding interaction on any allowed
pair. The phases cancel in the conjugation. Hence
\begin{equation}
 \left\langle e^{-itH_a}:t\in\mathbb R,\ 1\leq a\leq N\right\rangle
 =G.
 \label{eq:adjacent-and-all-pairs-generate-same-group}
\end{equation}

For small \(t\), a pulse has the expansion \(I-itH+O(t^2)\),
so its infinitesimal generator is the anti-Hermitian matrix \(-iH\). Real linear combinations and commutators of these matrices form the Lie algebra used to test which gates are available. The following lemma connects it to finite pulse products.

\begin{lemma}[The control group and its Lie algebra]
\label{lem:pulse-group-closed}
The group \(G\) is compact and connected, and
\begin{equation}
 \operatorname{Lie}(G)=\mathfrak L,
 \qquad
 \mathfrak L:=\operatorname{Lie}_{\mathbb R}\{-iH_1,\ldots,-iH_N\}.
 \label{eq:physical-dynamical-lie-algebra}
\end{equation}
\end{lemma}

\begin{proof}
An exchange pulse has period \(2\pi\), and a pulse generated by \(K=3P_\Omega\) has period \(2\pi/3\). Each pulse family is thus a compact connected subgroup. A subgroup generated by
finitely many such subgroups of a compact Lie group is compact and connected~\cite[Supplementary Note~1, Fact~1]{Marvian2022}. The generators belong to \(\operatorname{Lie}(G)\), and every
pulse product belongs to the connected analytic subgroup with Lie algebra \(\mathfrak L\). The two Lie algebras coincide.
\end{proof}

The algebra \(\mathfrak L\) is obtained by taking real linear combinations and repeated commutators of the generators $\{-iH_a\}_{a=1,\dots,N}$ until no new independent matrices appear. The block convention above allows this calculation to use either
the full physical matrices or their logical blocks. The lemma applies to the group of finite pulse products itself, so the control criterion below establishes exact implementations.

\subsection{The criterion for independent logical control}
\label{subsec:lie-algebra-test}

The example of locked qutrit subsystems in Sec.~\ref{subsec:proof-assembly}
shows why full action on every block is insufficient. Independent logical control requires every traceless direction in each block
to be available independently. In the following, we state this requirement and give an equivalent dimension test.

Take \(\mathfrak D=[\mathfrak L,\mathfrak L]\) as the real linear span of commutators between elements of \(\mathfrak L\). Every commutator is traceless in each logical block, so
\begin{equation}
 \mathfrak D\subseteq
 \bigoplus_\lambda\mathfrak{su}(W_\lambda).
 \label{eq:derived-contained-logical-algebra}
\end{equation}
Here \(\mathfrak{su}(W)\) consists of the traceless
anti-Hermitian matrices on \(W\). It is the Lie algebra of \(SU(W)\).
A logical subsystem of dimension \(m\) contributes \(m^2-1\)
real directions. The gauge dimensions do not enter this count
because the gates act as the identity on those spaces.
A space of dimension one contributes zero.

\begin{proposition}[Criterion for independent logical control]
\label{prop:global-lie-rank-criterion}
The system admits independent logical control if and only if
\begin{equation}
 \mathfrak D=\bigoplus_{\lambda\in\Lambda}
 \mathfrak{su}(W_\lambda).
 \label{eq:independent-control-algebra}
\end{equation}
Since the left-hand side is already contained in the right-hand
side, this is equivalent to the dimension test
\begin{equation}
 \dim_{\mathbb R}\mathfrak D
 =\sum_{\lambda\in\Lambda}(m_\lambda^2-1).
 \label{eq:global-lie-rank-dimension}
\end{equation}
\end{proposition}

\begin{proof}
If Eq.~\eqref{eq:independent-control-algebra} holds, the Lie algebra of the target group lies inside \(\mathfrak L=\operatorname{Lie}(G)\). Every target unitary has a traceless anti-Hermitian logarithm
on each logical block. Exponentiating these logarithms inside \(G\) gives the whole target group. Since \(G\) consists of finite pulse products, every element
of the target group has an exact finite pulse implementation.

Conversely, if \(G\) contains the target group
\(\prod_\lambda SU(W_\lambda)\), then
\(\mathfrak L\) contains  \(\bigoplus_\lambda\mathfrak{su}(W_\lambda)\).
Each \(\mathfrak{su}(m)\), for \(m\geq2\), equals its own commutator algebra: $\mathfrak{su}(m)=[\mathfrak{su}(m),\mathfrak{su}(m)]$. Taking commutators shows that
\(\mathfrak D\) contains this entire direct sum.
Together with the opposite inclusion in
Eq.~\eqref{eq:derived-contained-logical-algebra}, this proves
Eq.~\eqref{eq:independent-control-algebra}.
\end{proof}

Using \(\mathfrak D\) excludes scalar phase directions that are
outside the target of the theorem. The result guarantees
\(\prod_\lambda SU(W_\lambda)\), without requiring independent
control over relative phases between sectors.
In later appendices, \(\mathfrak L_{r,s}\) and
\(\mathfrak D_{r,s}\) denote these algebras with the system
size indices restored. Appendix~\ref{subsec:exact-rank-procedure}
describes exact calculations in the small systems as independent
checks of their analytic control proofs.
\section{Logical spaces and their dimensions when a qutrit is added}
\label{app:branching-multiplicity}

The induction proof of Theorem~\ref{thm:global-semiuniversality}
adds one fundamental or dual qutrit and asks whether logical
control extends to the enlarged system. Each new logical
subsystem is assembled from at most three logical subsystems
of the original system, which we call its predecessors.
We first describe this decomposition, then identify the small
logical spaces that require separate treatment in the control proof.

\subsection{Adding a qutrit and extending existing gates}
\label{sec:dynkin-labels-and-single-factor-branching}

For the addition of one qutrit, write \(\mathcal H=\mathcal H_{r,s}\) and
\(\mathcal H'=\mathcal H_{r,s}\otimes F\), where \(F=V\) or
\(V^*\). Thus \(\mathcal H'\) is identified with
\(\mathcal H_{r+1,s}\) or \(\mathcal H_{r,s+1}\), respectively.
Write \(W_\gamma\) for a logical subsystem in the original
system and \(W'_\lambda\) for one in the enlarged system.
Their dimensions are \(m_\gamma\) and \(m'_\lambda\), respectively.
Size indices \(r,s\) remain explicit when comparing several systems.

The Dynkin labels \((p,q)\in\mathbb Z_{\geq0}^2\) identify the irreducible
representation \(U_{p,q}\). In particular, \(V=U_{1,0}\) and
\(V^*=U_{0,1}\), and taking the complex conjugate exchanges \(p\) with \(q\).
Adding one qutrit gives the two tensor-product Pieri rules
\begin{equation}
\begin{aligned}
 U_{p,q}\otimes V
 &\cong U_{p+1,q}\oplus U_{p-1,q+1}\oplus U_{p,q-1},\\
 U_{p,q}\otimes V^*
 &\cong U_{p,q+1}\oplus U_{p+1,q-1}\oplus U_{p-1,q}.
\end{aligned}
\label{eq:qutrit-branching-rules}
\end{equation}
Terms with a negative coordinate are omitted. Each remaining
term appears once. One way to see the first rule is to represent
\(U_{p,q}\) by a Young diagram with row lengths \((p+q,q,0)\).
Adding a box to an allowed row gives the terms above.
A full column of three boxes represents the determinant, which
is one on \(SU(3)\), so removing that column leaves the
representation unchanged. The second rule follows by taking
complex conjugates.

For a fixed enlarged sector \(\lambda=(p,q)\), let
\(\Gamma(\lambda)\) be the set of original sectors that
contribute to it. The candidates are
\begin{equation*}
\begin{array}{c|l}
 \text{added qutrit}&\text{candidate original labels}\\ \hline
 V&(p-1,q),\ (p+1,q-1),\ (p,q+1)\\
 V^*&(p,q-1),\ (p-1,q+1),\ (p+1,q)
\end{array}
\end{equation*}
Keep only labels present in \(\mathcal H_{r,s}\). When the size
or added type needs to be explicit, these sets are denoted by
\(\Gamma_{r,s}^{+}(\lambda)\) and
\(\Gamma_{r,s}^{-}(\lambda)\) for fundamental and dual qutrit addition, respectively.

Each original copy of a contributing representation supplies one
copy of \(U_\lambda\) after adding a qutrit. Consequently,
\begin{equation}
 W'_\lambda\cong
 \bigoplus_{\gamma\in\Gamma(\lambda)}W_\gamma,
 \qquad
 m'_\lambda=\sum_{\gamma\in\Gamma(\lambda)}m_\gamma.
 \label{eq:branching-logical-spaces}
\end{equation}
The direct sum is a decomposition of the logical space. It records
which sector of the original system each contribution comes from. For example, adding \(V^*\) to \(\mathcal H_{3,0}\) produces
the logical qutrit used in the main text:
\begin{equation*}
 W'_{(2,0)}\cong W_{(3,0)}\oplus W_{(1,1)}
 \cong\mathbb C\oplus\mathbb C^2.
\end{equation*}
The predecessor multiplicities are one and two, so the new multiplicity is three. These are two subspaces of the same logical qutrit encoded in $W'_{(2,0)}$. 

\paragraph{Choosing the logical spaces and their bases.}
\label{subsec:CG-branching-decomposition}
The formal construction of a logical space is
\begin{equation}
 \begin{aligned}
 W_{r,s}^{p,q}
 &=\operatorname{Hom}_{SU(3)}(U_{p,q},\mathcal H_{r,s}),\\
 m_{r,s}(p,q)&=\dim W_{r,s}^{p,q}.
 \end{aligned}
 \label{eq:multiplicity-space-definition}
\end{equation}
An element of this space is a linear map that intertwines
the irreducible action on \(U_{p,q}\) with the collective
action on \(\mathcal H_{r,s}\). For such maps \(S\) and \(T\),
Schur's lemma defines
the inner product by \(S^\dagger T=\langle S,T\rangle I\).
Choosing isometric embeddings in
Eq.~\eqref{eq:qutrit-branching-rules} makes the direct-sum identification in Eq.~\eqref{eq:branching-logical-spaces} unitary.
Fix compatible orthonormal bases for this identification.
The choice is not unique, but no particular Clebsch--Gordan
coefficients are required for this construction.

\paragraph{Extending existing operators.}
\label{sec:single-factor-extension-conventions}
Write \(X'\) for \(X\otimes I_F\) acting on the enlarged system. We keep all \(V\) factors before all \(V^*\) factors, so when we add a \(V\) factor, it is placed before the \(V^*\) factors. An added \(V^*\) factor is appended on the right. 
The extension preserves products, adjoints and invariance under the collective action.

In the chosen branching bases, the logical blocks of an invariant
operator \(X\) satisfy
\begin{equation}
 X'_\lambda=\bigoplus_{\gamma\in\Gamma(\lambda)}X_\gamma.
 \label{eq:branching-existing-action}
\end{equation}
Indeed, \(X\) acts only on the logical factor of each original sector. Decomposing its representation factor after adding the qutrit leaves that logical action unchanged.

If independent logical control is available before qutrit addition,
an arbitrary \(Y_\gamma\in SU(W_\gamma)\) can thus be chosen
for each predecessor, independently of the others.
The same \(Y_\gamma\) appears in every enlarged sector
containing a contribution from \(W_\gamma\).
Each inherited gate has determinant one on every enlarged logical
block. These gates preserve the separate contributions.
Interactions involving the added qutrit must connect them to
obtain all gates on \(W'_\lambda\).

For subsequent calculations, Eq.~\eqref{eq:branching-logical-spaces}
gives the recursions
\begin{equation}
\begin{aligned}
 m_{r+1,s}(p,q)
 &=m_{r,s}(p-1,q)+m_{r,s}(p+1,q-1)\\
 &+m_{r,s}(p,q+1),\\
 m_{r,s+1}(p,q)
 &=m_{r,s}(p,q-1)+m_{r,s}(p-1,q+1)\\
 &+m_{r,s}(p+1,q).
\end{aligned}
\label{eq:multiplicity-recursions}
\end{equation}
These hold for target Dynkin labels
\((p,q)\in\mathbb Z_{\geq0}^2\). An absent original label,
including one with a negative coordinate, has multiplicity zero.
The initial value is \(m_{0,0}(0,0)=1\), with all other
multiplicities zero. The recursions preserve the selection rule
\(p-q\equiv r-s\pmod3\). This also follows because the central
element \(\omega I_3\), with \(\omega=e^{2\pi i/3}\), acts as
\(\omega^{p-q}\) on \(U_{p,q}\) and as \(\omega^{r-s}\) on
\(\mathcal H_{r,s}\). These phases must be equal, so their exponents agree modulo three.

\subsection{Which symmetry sectors occur}
\label{sec:sector-support}

We first locate all the sectors, for any \(r,s\geq0\).
Coordinates relative to the highest label \((r,s)\) will also
simplify the later dimension calculation. Write
\begin{equation}
 (p,q)=(r,s)-a(2,-1)-b(-1,2).
 \label{eq:sector-coordinates}
\end{equation}
Thus \(p=r-2a+b\), \(q=s+a-2b\), or equivalently
\begin{equation*}
 a=\frac{2r+s-2p-q}{3},\qquad
 b=\frac{r+2s-p-2q}{3}.
\end{equation*}
The sector \((r,s)\) has coordinates \((0,0)\).
Dynkin labels \((p,q)\) specify the irreducible representation
\(U_{p,q}\), while the coordinates \((a,b)\) record how labels
vary when a fundamental or dual qutrit is added.
Adding \(V\) gives the candidate coordinates
\((a,b),(a+1,b),(a+1,b+1)\), while adding \(V^*\) gives
\((a,b),(a,b+1),(a+1,b+1)\). Only candidates with nonnegative
Dynkin labels are retained. These steps start at \((0,0)\) and
never decrease either coordinate. Thus every sector that appears
has integer coordinates \(a,b\geq0\).

Let
\begin{equation*}
 \mathscr D_{r,s}
 =\{(a,b)\in\mathbb Z_{\geq0}^2:
       r-2a+b\geq0,\ s+a-2b\geq0\}
\end{equation*}
be the set of admissible coordinate pairs. On this set, write
\(M_{r,s}(a,b):=m_{r,s}(r-2a+b,s+a-2b)\), and set
\(M_{r,s}=0\) outside it. For an admissible target pair in the enlarged system, the recursions become
\begin{equation}
\begin{aligned}
 M_{r+1,s}(a,b)
 &=M_{r,s}(a,b)+M_{r,s}(a-1,b)\\
 &+M_{r,s}(a-1,b-1),\\
 M_{r,s+1}(a,b)
 &=M_{r,s}(a,b)+M_{r,s}(a,b-1)\\
 &+M_{r,s}(a-1,b-1).
\end{aligned}
\label{eq:coordinate-multiplicity-recursions}
\end{equation}
The first identity requires \((a,b)\in\mathscr D_{r+1,s}\),
and the second requires \((a,b)\in\mathscr D_{r,s+1}\).

\begin{lemma}[Which sectors occur]
\label{lem:sector-support}
For every \(r,s\geq0\), the occurring sectors are exactly
those whose coordinates belong to \(\mathscr D_{r,s}\).
\end{lemma}

\begin{proof}
We have already shown that every occurring sector has admissible
coordinates. For the converse, start with
\(\mathscr D_{0,0}=\{(0,0)\}\), where the sole pair has
multiplicity one. It is enough to show that each admissible
coordinate pair in the enlarged system has an admissible
original pair among the terms of its recursion. 

For addition of \(V\), let \((p,q)\) be the enlarged Dynkin label.
Choose the original pair as follows:
\begin{equation}
\begin{array}{c|c|c}
 \text{target condition}&\text{original pair}&\text{original label}\\ \hline
 p>0&(a,b)&(p-1,q)\\
 p=0<q&(a-1,b)&(1,q-1)\\
 p=q=0&(a-1,b-1)&(0,1).
\end{array}
\label{eq:support-contribution}
\end{equation}
All displayed original labels are nonnegative. In the last two
rows, \(p=0\) gives \(2a=r+1+b\geq1\), hence \(a\geq1\).
In the last row, \(q=0\) also gives \(2b=s+a\geq1\), hence
\(b\geq1\). Thus the original coordinates are nonnegative as
well. Exchanging \(r,s\), \(a,b\) and \(p,q\) gives the rule
for addition of \(V^*\).
By induction, the chosen original pair has positive multiplicity.
The recursion is a sum of nonnegative terms, so the enlarged
pair also has positive multiplicity.
\end{proof}

The same choice preserves positive coordinates when
\(a,b>0\) and the original system has \(r\geq3,s\geq1\).
For addition of \(V\), a coordinate is reduced only if
\(p=0\). Then \(2a=r+1+b\geq5\), so \(a\geq3\).
If also \(q=0\), then \(2b=s+a\geq4\), so \(b\geq2\).
For addition of \(V^*\), a coordinate is reduced only if
\(q=0\). Then \(2b=s+1+a\geq3\), so \(b\geq2\).
If also \(p=0\), then \(2a=r+b\geq5\), so \(a\geq3\).
Subtracting one in the corresponding cases leaves both coordinates
positive. We use this fact below to propagate a dimension bound.

\subsection{Logical spaces of dimension one or two}
\label{sec:low-multiplicity-classification}

The general control argument in
Lemma~\ref{lem:regular-sector-controllability} applies when an
enlarged logical space contains an original logical space of
dimension at least three. To find the remaining cases, we now
locate the original spaces of dimensions one and two.
The classification in this and the next subsection assumes
\(r\geq3\) and \(s\geq1\).
On the coordinate axes, a short list of initial values identifies
all dimensions one and two. When both coordinates are positive,
an induction starting at \((r,s)=(3,1)\) gives a lower bound of three.

\begin{lemma}[Logical spaces of dimension one or two]
\label{lem:low-multiplicities}
Let \(r\geq3\) and \(s\geq1\). The logical spaces of dimensions one and two are exactly
those in Table~\ref{tab:small-logical-spaces}.
All other logical spaces have dimension at least three.
If both coordinates are positive and \(s\geq2\), their dimension
is at least five.
\end{lemma}

\begin{table}[tb]
\centering
\begin{tabular}{c c c c}
\hline\hline
\((a,b)\)&Condition&Sector \((p,q)\)&Dimension\\
\hline
\((0,0)\)&All \(r,s\)&\((r,s)\)&1\\
\((0,1)\)&\(s=2\)&\((r+1,0)\)&1\\
\((1,0)\)&\(r=3\)&\((1,s+1)\)&2\\
\((2,0)\)&\(r=4\)&\((0,s+2)\)&2\\
\((0,1)\)&\(s=3\)&\((r+1,1)\)&2\\
\((0,2)\)&\(s=4\)&\((r+2,0)\)&2\\
\hline\hline
\end{tabular}
\caption{All logical spaces of dimension one or two for
\(r\geq3\) and \(s\geq1\). Conditions in the second column
are additional restrictions within this range.}
\label{tab:small-logical-spaces}
\end{table}

\begin{proof}
We calculate the dimensions on the two coordinate axes and
then bound them when both coordinates are positive.

\paragraph{The coordinate axes.}
When \(b=0\), the second recursion in
Eq.~\eqref{eq:coordinate-multiplicity-recursions} shows that
the multiplicity is independent of \(s\). Only the \(r\)
fundamental qutrits affect its value. Write
\(B_n(k)=M_{n,0}(k,0)\) for \(0\leq2k\leq n\), and set
\(B_n(k)=0\) otherwise. Exchanging the qutrit types gives
\begin{equation}
 M_{r,s}(a,0)=B_r(a),\qquad
 M_{r,s}(0,b)=B_s(b).
 \label{eq:axis-multiplicities}
\end{equation}
These formulas hold throughout their admissible ranges,
including systems with only one qutrit type. The first recursion in
Eq.~\eqref{eq:coordinate-multiplicity-recursions} becomes
\begin{equation*}
 B_{n+1}(k)=B_n(k)+B_n(k-1),
 \qquad 0\leq2k\leq n+1.
\end{equation*}
Starting from \(B_0(0)=1\), it gives
\begin{equation*}
\begin{array}{c|ccc}
 n&B_n(0)&B_n(1)&B_n(2)\\ \hline
 1&1&&\\
 2&1&1&\\
 3&1&2&\\
 4&1&3&2\\
 5&1&4&5
\end{array}
\end{equation*}
where blank entries are outside the admissible range.
For every \(n\geq5\), each admissible entry with \(k>0\)
is at least three. To propagate this bound from row \(n\)
to row \(n+1\), use \(B_n(k)\) when \(2k\leq n\).
Otherwise \(2k=n+1\), and \(B_n(k-1)\) supplies the bound,
since \(k-1>0\) is admissible. The recursion also gives
\(B_n(0)=1\) for all \(n\), and \(B_n(1)=n-1\) for
\(n\geq2\). The entries one and two on the axes are thus
exactly those in Table~\ref{tab:small-logical-spaces}.

\paragraph{Pairs with both coordinates positive.}
At \((r,s)=(3,1)\), the admissible pairs with both
coordinates positive are \((1,1)\) and \((2,1)\).
Iterating Eq.~\eqref{eq:multiplicity-recursions} from
\(m_{0,0}(0,0)=1\) gives multiplicity three for each.
The choice in
Eq.~\eqref{eq:support-contribution} preserves positive coordinates
in every subsequent enlargement. If the lower bound of three holds
in the original system, each enlarged pair thus inherits a
contribution of dimension at least three. The remaining terms in
the recursion are nonnegative. Starting from \((3,1)\), induction
gives, for every \(r\geq3\) and \(s\geq1\),
\begin{equation}
 M_{r,s}(a,b)\geq3
 \quad\text{for }(a,b)\in\mathscr D_{r,s},\ a,b>0.
 \label{eq:interior-dimension-at-least-three}
\end{equation}
Together with the calculation on the axes, this proves the classification.

For the stronger bound when \(s\geq2\), start instead at
\((3,2)\). Its positive coordinate pairs are
\((1,1),(2,1),(2,2)\), with multiplicities \(6,5,6\).
The same argument gives, for every \(r\geq3\) and \(s\geq2\),
\begin{equation}
 M_{r,s}(a,b)\geq5
 \quad\text{for }(a,b)\in\mathscr D_{r,s},\ a,b>0.
 \label{eq:interior-lower-bound-conclusion}
\end{equation}
\end{proof}

\subsection{Enlarged spaces assembled from low-dimensional logical subsystems}
\label{sec:source-multiplicity-at-most-two}

For an original system with \(r\geq3\) and \(s\geq1\),
only three dimension patterns remain when an enlarged logical
space has dimension greater than one and no predecessor
of dimension at least three. The corresponding control arguments
are given in Appendix~\ref{app:sectorwise-separation}. 

\begin{proposition}[Enlarged spaces with small contributions]
\label{prop:small-contributions}
Let \(r\geq3\), \(s\geq1\), and add either type of qutrit.
The enlarged sectors satisfying
\begin{equation}
 m'_\lambda>1,
 \qquad m_\gamma\leq2
 \quad\text{for every }\gamma\in\Gamma(\lambda)
 \label{eq:small-contribution-condition}
\end{equation}
have contribution dimensions \((2)\), \((2,1)\), or
\((1,1)\). Thus the enlarged dimension is two or three.
Table~\ref{tab:small-contributions} gives all nine families of
sector labels for these three patterns.
\end{proposition}

\begin{table*}[t]
\centering
\begin{tabular}{c c c c c}
\hline\hline
Added qutrit&Condition&Enlarged label \(\lambda\)
&Original labels \(\Gamma(\lambda)\)&Dimensions\\
\hline
\(V\)&\(r=3\)&\((0,s+2)\)&\(\{(1,s+1)\}\)&\(2\)\\
\(V\)&\(r=3\)&\((2,s+1)\)&\(\{(1,s+1),(3,s)\}\)&\(2+1\)\\
\(V\)&\(s=3\)&\((r+2,1)\)&\(\{(r+1,1)\}\)&\(2\)\\
\(V\)&\(s=4\)&\((r+3,0)\)&\(\{(r+2,0)\}\)&\(2\)\\
\hline
\(V^*\)&\(s=2\)&\((r+1,1)\)&\(\{(r+1,0),(r,2)\}\)&\(1+1\)\\
\(V^*\)&\(s=3\)&\((r+1,2)\)&\(\{(r+1,1),(r,3)\}\)&\(2+1\)\\
\(V^*\)&\(s=3\)&\((r+2,0)\)&\(\{(r+1,1)\}\)&\(2\)\\
\(V^*\)&\(r=3\)&\((1,s+2)\)&\(\{(1,s+1)\}\)&\(2\)\\
\(V^*\)&\(r=4\)&\((0,s+3)\)&\(\{(0,s+2)\}\)&\(2\)\\
\hline\hline
\end{tabular}
\caption{The sector families in
Proposition~\ref{prop:small-contributions}, for \(r\geq3\)
and \(s\geq1\). The last column gives the dimensions of the
contributions in the order displayed. Their sum is the
dimension of the enlarged logical space. Each row applies
whenever its condition holds, including when several rows apply.}
\label{tab:small-contributions}
\end{table*}

\begin{proof}
If the enlarged coordinates satisfy \(a,b>0\),
Eq.~\eqref{eq:support-contribution} and its counterpart
for \(V^*\) give an original contribution with both coordinates positive.
Its dimension is at least three by
Lemma~\ref{lem:low-multiplicities}, contradicting
Eq.~\eqref{eq:small-contribution-condition}. Only the axes
remain.

On an axis, adding a qutrit either leaves the multiplicity
equal to \(B_n(k)\), or changes it to
\(B_n(k)+B_n(k-1)\). Adding \(V\) increases the parameter
on the \(b=0\) axis, with \(n=r,k=a\), and leaves the
\(a=0\) axis unchanged. Adding \(V^*\) increases the
parameter on the \(a=0\) axis, with \(n=s,k=b\), and
leaves the other axis unchanged.

For an unchanged axis, the single contribution must have
dimension two. The values of \(B_n(k)\) found above give
only \((n,k)=(3,1)\) and \((4,2)\).

For an increasing parameter, target admissibility requires
\(0\leq2k\leq n+1\). The sum must exceed one and each
nonzero term must be at most two. The sum condition excludes
\(k=0\). For \(n\geq5\), the axis argument above gives a
contribution of dimension at least three, so only \(n=1,2,3,4\)
remain. At \(n=1\), both admissible targets have dimension one.
Inspecting the other three rows leaves
\begin{equation*}
\begin{array}{c|c|c}
 (n,k)&B_n(k)+B_n(k-1)&\text{enlarged dimension}\\ \hline
 (2,1)&1+1&2\\
 (3,1)&2+1&3\\
 (3,2)&0+2&2
\end{array}
\end{equation*}
The last row has only one nonzero contribution, of dimension two. For addition of \(V\), \(n=r\geq3\) excludes the first
row. For addition of \(V^*\), \(n=s\geq1\) allows it. Substituting the resulting coordinates into
Eq.~\eqref{eq:sector-coordinates} gives exactly
Table~\ref{tab:small-contributions}.
\end{proof}

These patterns leave two control arguments in
Appendix~\ref{app:sectorwise-separation}. The lists \((2)\)
and \((1,1)\) give a two-dimensional enlarged space, where
irreducibility alone will imply full \(SU(2)\) control.
The list \((2,1)\) requires coupling the two-dimensional
subspace with an existing \(SU(2)\) action to the remaining
one-dimensional subspace, giving full \(SU(3)\) control. After establishing control within each
enlarged sector, a further argument establishes independence between different sectors.

\section{Extending logical control and separating sectors}
\label{app:sectorwise-separation}

After adding either type of qutrit, we must establish full special-unitary
control within each enlarged logical space and show that the gates can be
chosen independently across sectors. We first give the general control
arguments, then use the branching classification to establish both
properties in the range needed for the induction proving
Theorem~\ref{thm:global-semiuniversality}.
Appendix~\ref{app:propagation-base} applies these arguments to the initial
systems and completes the proof.

We use the notation of Appendix~\ref{app:branching-multiplicity}.
The predecessors \(W_\gamma\) contribute to an enlarged space
\(W'_\lambda\), with predecessor labels \(\Gamma(\lambda)\).
Let \(G'\) denote the pulse group generated by the available two-body
interactions on the enlarged system. Independent logical control on the
original system \((r,s)\) supplies \(SU(W_\gamma)\) on each contribution
\(W_\gamma\) of \(W'_\lambda\), with the identity on the other
contributions within \(W'_\lambda\), by
Eq.~\eqref{eq:branching-existing-action}. The full set of interactions on
the enlarged system acts irreducibly on \(W'_\lambda\), since its
associative algebra is the full matrix algebra on that space.

An original logical space can contribute to several enlarged sectors,
so its gate appears in each of them. Independence between enlarged
sectors must be established separately.

\subsection{Extending control from a contributing space}
\label{subsec:propagation-control-inputs}

Irreducibility supplies a coupling between any proper subspace on which
gates are available and its complement. When that subspace has dimension at least three, two diagonal commutators select a single new direction from the coupling. The available gates then connect this direction to every
vector in the starting subspace, extending full control by one dimension.

\begin{lemma}[Extending control from a subspace]
\label{lem:propagation-irreducible-extension}
Let \(\mathfrak j\subseteq\mathfrak u(\mathcal K)\) be a real
matrix Lie algebra acting irreducibly on the complex Hilbert space
\(\mathcal K\). If it contains
\(\mathfrak{su}(\mathcal K_0)\oplus0_{\mathcal K_0^\perp}\)
for a subspace \(\mathcal K_0\) of dimension at least three, then
\begin{equation}
 \mathfrak{su}(\mathcal K)\subseteq\mathfrak j.
 \label{eq:propagation-extension-full-algebra}
\end{equation}
\end{lemma}

\begin{proof}
Choose a subspace \(M\supseteq\mathcal K_0\) of largest dimension
such that \(\mathfrak{su}(M)\), acting as zero on \(M^\perp\),
is contained in \(\mathfrak j\). If \(M\neq\mathcal K\),
irreducibility supplies an \(X\in\mathfrak j\) with a nonzero
block \(B:M^\perp\to M\). Subtract the traceless part of its
\(M\)-diagonal block, which is already available in the supported
\(\mathfrak{su}(M)\). In the decomposition \(M\oplus M^\perp\),
the resulting matrix has the form
\begin{equation*}
 X=\begin{pmatrix}i\alpha I_M&B\\-B^\dagger&Z\end{pmatrix},
 \qquad B\neq0,\quad \alpha\in\mathbb R.
\end{equation*}

Choose an orthonormal basis of \(M\) and a nonzero row \(j\) of
\(B\). Since \(\dim M\geq3\), there are two distinct indices
\(k,\ell\), both different from \(j\). The available matrices
\(D_1=i(E_{jj}-E_{kk})\) and \(D_2=i(E_{jj}-E_{\ell\ell})\),
extended by zero on \(M^\perp\), commute with the diagonal blocks
of \(X\). Here \(E_{ab}:=\ket{e_a}\bra{e_b}\).
The first commutator removes the diagonal blocks, and the second
selects the chosen row of \(B\). Since \(D_1D_2=-E_{jj}\) on
\(M\),
\begin{equation}
 [D_1,[D_2,X]]
 =-\begin{pmatrix}0&e_jw^\dagger\\-we_j^\dagger&0\end{pmatrix},
 \qquad w=B^\dagger e_j\neq0.
 \label{eq:extension-isolated-direction}
\end{equation}
After real rescaling, this couples \(e_j\) to the unit vector
\(f=w/\|w\|\) and vanishes elsewhere. Commutation with \(D_1\)
gives the coupling with \(e_j\) replaced by \(ie_j\).
Commutators with the existing \(\mathfrak{su}(M)\) give both
coupling directions between \(f\) and every basis vector of \(M\).
Their brackets also supply \(i(E_{jj}-P_f)\), where \(P_f\) is
the projector onto \(\mathbb Cf\).

For \(m=\dim M\), these are the \(2m\) real coupling directions
and one additional traceless diagonal direction. Together with the
\(m^2-1\) directions already available on \(M\), they span
\(\mathfrak{su}(M\oplus\mathbb Cf)\), acting as zero on its
orthogonal complement. This contradicts the maximality of \(M\).
Thus \(M=\mathcal K\), which proves the required containment.
\end{proof}

Lemma~\ref{lem:propagation-irreducible-extension} is also the Lie-algebra form of
Ref.~\cite[Lemma~2 and Remark~4]{HulseLiuMarvian2024}.
The dimension assumption supplies the three distinct indices used
above. The dimension assumption cannot be removed: a supported
\(\mathfrak{su}(2)\) can lie in a proper irreducible compact symplectic
algebra. We treat two-dimensional contributions separately below.

% \paragraph{Application to an enlarged logical space.}
% \label{subsec:large-source-sector-controllability}

\begin{lemma}[Contributions of dimension at least three]
\label{lem:regular-sector-controllability}
Let \(r,s\geq1\), assume independent logical control on
\(\mathcal H_{r,s}\), and add either type of qutrit. If an enlarged
logical space \(W'_\lambda\) has a predecessor
\(W_\gamma\) of dimension at least three, then
\begin{equation}
 [G']_\lambda\supseteq SU(W'_\lambda).
 \label{eq:regular-sector-controllability-conclusion}
\end{equation}
\end{lemma}

\begin{proof}
In the branching basis, independent control on the original system gives
\begin{equation*}
 SU(W_\gamma)\oplus I_{W_\gamma^\perp}
 \subseteq [G']_\lambda,
\end{equation*}
where the complement is taken within \(W'_\lambda\).
The block group is compact and connected by
Lemma~\ref{lem:pulse-group-closed}. Its Lie algebra contains
\(\mathfrak{su}(W_\gamma)\) with zero action on that complement,
and acts irreducibly by the full associative-algebra property above.
Lemma~\ref{lem:propagation-irreducible-extension} gives
\(\mathfrak{su}(W'_\lambda)\subseteq\Lie([G']_\lambda)\).
Exponentiating inside the block group proves the result.
\end{proof}

\subsection{Completing control in dimensions two and three}
\label{subsec:elementary-control}

We now treat two cases: a two-dimensional predecessor coupled to a
line, and an entire logical space of dimension two. In the latter case,
irreducibility alone gives full special-unitary control. Both arguments
apply independently of any branching classification.

\begin{lemma}[A two-dimensional space coupled to a line]
\label{lem:two-one-control}
Let \(T\subseteq U(E\oplus L)\) be a connected closed group,
where \(\dim E=2\) and \(\dim L=1\). If its Lie algebra acts
irreducibly and \(T\) contains \(SU(E)\oplus I_L\), then
\(T\supseteq SU(E\oplus L)\).
\end{lemma}

\begin{proof}
Write \(\mathfrak t=\Lie(T)\). Irreducibility supplies a Hermitian
matrix \(A\) with \(-iA\in\mathfrak t\) and a nonzero block
between \(E\) and \(L\). The available gate
\(D=(-I_E)\oplus I_L\) belongs to \(SU(E)\oplus I_L\),
since \(\det(-I_E)=1\). Its conjugation reverses the sign of the
coupling while preserving the diagonal blocks. Hence
\begin{equation}
 C=\frac{A-DAD^\dagger}{2}
   =\begin{pmatrix}0&v\\v^\dagger&0\end{pmatrix},
 \qquad v\neq0,
 \label{eq:two-one-isolated-coupling}
\end{equation}
and \(-iC\in\mathfrak t\).

The existing \(SU(E)\) gates rotate \(v\) through the sphere of
radius \(\|v\|\) in \(E\). The real linear span of the resulting
conjugates of \(-iC\) contains all four coupling directions
\begin{equation*}
 X(w)=\begin{pmatrix}0&w\\-w^\dagger&0\end{pmatrix},
 \qquad w\in E.
\end{equation*}
For a unit vector \(e\in E\),
\begin{equation*}
 [X(e),X(ie)]=2i\bigl(ee^\dagger\oplus(-1)\bigr).
\end{equation*}
This gives a traceless diagonal direction with nonzero action on
\(L\). The three old \(\mathfrak{su}(E)\) directions, four
coupling directions and this diagonal direction are the eight
directions of \(\mathfrak{su}(E\oplus L)\).
Exponentiating inside \(T\) proves the claim.
\end{proof}

\begin{lemma}[Irreducible control in dimension two]
\label{lem:two-lines-control}
Let \(\mathfrak j\subseteq\mathfrak u(E)\) be a real matrix
Lie algebra acting irreducibly on a complex two-dimensional space
\(E\). Then \(\mathfrak{su}(E)\subseteq\mathfrak j\).
\end{lemma}

\begin{proof}
Scalar matrices alone cannot act irreducibly, so there is a
nonscalar Hermitian \(A\) with \(-iA\in\mathfrak j\).
In an eigenbasis, write
\(A=\alpha I+\delta\sigma_z\), where
\(\alpha,\delta\in\mathbb R\) and \(\delta\neq0\).
Irreducibility supplies another Hermitian \(B\), with
\(-iB\in\mathfrak j\), that mixes the two eigenvectors.
Adjusting one basis phase gives
\begin{equation*}
 B=\begin{pmatrix}u&t\\t&v\end{pmatrix},
 \qquad u,v\in\mathbb R,\quad t>0,
\end{equation*}
without changing \(A\). The Pauli commutators give
\begin{equation}
\begin{aligned}
 [-iA,-iB]&=-2i\delta t\sigma_y,\\
 [-iA,[-iA,-iB]]&=4i\delta^2t\sigma_x.
\end{aligned}
 \label{eq:irreducible-qubit-commutators}
\end{equation}
Their commutator supplies \(i\sigma_z\), so all three directions
of \(\mathfrak{su}(E)\) are available. The scalar part of
\(A\) cancels in the commutators without being removed from the
physical Hamiltonian.
\end{proof}

Exponentiating gives \(SU(E)\) in the corresponding connected
block group. This lemma applies when the entire space has dimension
two. It does not extend a supported \(SU(2)\) action to an
arbitrary larger irreducible space.

\subsection{Choosing gates independently in different sectors}

Full control within each logical space can still leave gates on two
sectors linked, as the example in Sec.~\ref{subsec:proof-assembly}
shows. Such links are possible only between sectors whose logical spaces
have equal dimension.

\begin{lemma}[Independent gates across sectors]
\label{lem:propagation-semiuniversality-criterion}
Let \(T\) be a connected closed group of invariant unitaries, with
logical spaces \(W_\lambda\) as in Appendix~\ref{app:mathematical-framework}. Write \(T_\lambda\)
for its group of logical blocks on \(W_\lambda\). Suppose that
\(T_\lambda\supseteq SU(W_\lambda)\) for every sector. Suppose also
that, for any distinct sectors whose logical spaces have the same dimension,
at least two,
some \(V\in T\) acts as the identity on one and is nonscalar on
the other. Then
\begin{equation}
 T\supseteq
 \left\{\bigoplus_\lambda I_{U_\lambda}\otimes Y_\lambda:
 Y_\lambda\in SU(W_\lambda)\right\}.
 \label{eq:propagation-semiuniversality-criterion-output}
\end{equation}
\end{lemma}

\begin{proof}
Let \(\mathfrak t=\Lie(T)\) and
\(\mathfrak d=[\mathfrak t,\mathfrak t]\). The projection of
\(\mathfrak d\) onto each nontrivial logical block is
\(\mathfrak{su}(W_\lambda)\). The compact semisimple algebra
\(\mathfrak d\) is a direct sum of commuting simple ideals.
These are the factors that can be varied independently.
Because the total projection onto a logical block is surjective, the image of each ideal is an ideal in \(\mathfrak{su}(W_\lambda)\).
That image is either zero or the whole block algebra. Two commuting
ideals cannot both map onto the same noncommutative block, so exactly
one does. Its kernel is zero by simplicity, making the projection
an isomorphism. A failure of
independence can thus occur only if one ideal acts on two sectors,
whose logical spaces must have the same dimension.

For such a pair, their joint projection would have the form
\(\{(X,\phi(X)):X\in\mathfrak{su}(n)\}\), where \(\phi\) is a
Lie-algebra isomorphism. By hypothesis, an available gate has blocks
\((A,I)\), with \(A\) nonscalar. Conjugation by this gate preserves
\(\mathfrak d\), so
\begin{equation*}
 \phi(AXA^\dagger)=\phi(X)
 \qquad\text{for every }X\in\mathfrak{su}(n).
\end{equation*}
Injectivity of \(\phi\) makes \(A\) commute with all of
\(\mathfrak{su}(n)\), so \(A\) is scalar, a contradiction.
It follows that
\(\mathfrak d=\bigoplus_{\lambda:\dim W_\lambda\geq2}
\mathfrak{su}(W_\lambda)\). Exponentiating inside \(T\) proves
the conclusion. One-dimensional spaces have \(SU(1)=\{I\}\).
\end{proof}

This criterion is also given in Ref.~\cite[Lemma~1]{HulseLiuMarvian2024}.

\subsection{Applying the small-dimension classification}
\label{subsec:thin-sector-controllability}

We now apply the preceding lemmas when \(r\geq3,s\geq1\).
Proposition~\ref{prop:small-contributions} leaves only the lists of
predecessor dimensions
\((2)\), \((2,1)\) and \((1,1)\) when an enlarged logical space
has dimension greater than one and every contribution has dimension at most two. The first and last lists give a two-dimensional space. The middle list requires extending the gates on a two-dimensional contribution to an extra line.

\begin{lemma}[Control from contributing spaces of dimension at most two]
\label{lem:thin-sector-controllability}
Let \(r\geq3,s\geq1\) and assume
\(G_{r,s}\supseteq\mathcal U_{\mathrm{log}}(r,s)\).
For either addition, let \(\lambda\) be an enlarged sector with
\(m'_\lambda>1\) and \(m_\gamma\leq2\) for every
\(\gamma\in\Gamma(\lambda)\). Then
\begin{equation}
 [G']_\lambda\supseteq SU(W'_\lambda).
 \label{eq:small-contribution-controllability}
\end{equation}
\end{lemma}

\begin{proof}
The block Lie algebra acts irreducibly by
Eq.~\eqref{eq:interaction-commutant}. For \((2)\) or \((1,1)\),
the enlarged space has dimension two, so
Lemma~\ref{lem:two-lines-control} gives its full \(SU(2)\) action.
For \((2,1)\), independent original control supplies the supported
\(SU(2)\) action on the two-dimensional contribution.
Lemma~\ref{lem:two-one-control} then gives \(SU(3)\) on the
enlarged space.
\end{proof}

\subsection{Distinguishing enlarged sectors with an existing gate}
\label{sec:branching-separation}

To apply Lemma~\ref{lem:propagation-semiuniversality-criterion},
we find a predecessor of dimension at least two that
contributes to exactly one of two enlarged sectors whose logical spaces
have equal dimension. Independent control on the original system supplies
a nonscalar gate on that contribution and the identity on all other
original logical spaces. Equation~\eqref{eq:branching-existing-action}
then gives the required actions on the enlarged sectors.

\begin{lemma}[Branching separation]
\label{lem:branching-separation}
Let \(r\geq3,s\geq1\). For either addition, let
\(\lambda\neq\kappa\) be enlarged sectors
with \(m'_\lambda=m'_\kappa\geq2\). Then an original sector
\(\gamma\) of multiplicity \(m_\gamma\geq2\) contributes to
exactly one of them. After interchanging \(\lambda\) and \(\kappa\)
if necessary,
\begin{equation}
 \gamma\in\Gamma(\lambda)\setminus\Gamma(\kappa),
 \qquad m_\gamma\geq2.
 \label{eq:exclusive-contributing-sector}
\end{equation}
\end{lemma}

\begin{proof}
We first show that distinct enlarged sectors whose logical spaces have
dimension at least two have different sets of contributing sectors. We then use the
classification of one-dimensional spaces to find a contribution
on which a nontrivial gate is available.

\paragraph{Recovering a sector from its contributions.}
The three increments of a Dynkin label when adding \(V\) are
\((1,0)\), \((-1,1)\), and \((0,-1)\). Their six nonzero ordered
differences are distinct. If two enlarged sectors share a contributing
sector, the difference of their labels thus determines the ordered
pair of increments, and hence the shared contributing sector itself.
The increments for \(V^*\) are the negatives of these three vectors.
For either addition, it follows that
\begin{equation}
 |\Gamma(\lambda)\cap\Gamma(\kappa)|\leq1
 \qquad(\lambda\neq\kappa).
 \label{eq:branching-at-most-one-common}
\end{equation}
This observation uses only the tensor-product rule and holds for
all \(r,s\geq0\). Identical contributing sets for distinct enlarged
sectors could only be singletons.

To check singletons, use the coordinates of
Appendix~\ref{app:branching-multiplicity}. For an admissible enlarged
pair \((a,b)\), the contributing coordinate sets are
\begin{equation}
\begin{aligned}
 P^+(a,b)&=\{(a,b),(a-1,b),(a-1,b-1)\}\cap\mathscr D_{r,s},\\
 P^-(a,b)&=\{(a,b),(a,b-1),(a-1,b-1)\}\cap\mathscr D_{r,s}.
\end{aligned}
 \label{eq:contributing-coordinate-sets}
\end{equation}
Every point of \(\mathscr D_{r,s}\) occurs in the original system
by Lemma~\ref{lem:sector-support}.

For addition of \(V\), write the enlarged label as
\(p=r+1-2a+b\), \(q=s+a-2b\). The three candidates in
\(P^+(a,b)\) occur under exactly the following conditions:
\begin{equation}
\begin{array}{c|c}
 \text{original coordinates}&\text{condition}\\ \hline
 (a,b)&p\geq1\\
 (a-1,b)&a\geq1,\ q\geq1\\
 (a-1,b-1)&a,b\geq1
\end{array}
 \label{eq:branching-predecessor-presence}
\end{equation}
These are precisely the conditions that the original coordinates
and Dynkin labels be nonnegative. For \(r,s\geq1\),
exactly one candidate is present in the following cases:
\begin{equation}
\begin{array}{c|c}
 \text{singleton condition}&\text{sole original label}\\ \hline
 a=0 &(r+b,s-2b)\\
 b=p=0 &(1,q-1)\\
 p=q=0 &(0,1)
\end{array}
 \label{eq:branching-singleton-contributions}
\end{equation}
The first coordinates of the original labels are at least \(r\),
one and zero, respectively. When \(r\geq2\), these ranges are
disjoint, so the original label determines the row and hence the
enlarged sector.

When \(r=1\), the first two rows can coincide only at \(b=0\),
with original label \((1,s)\). This is the highest sector and
has multiplicity one. Its two singleton targets, \((2,s)\) and
\((0,s+1)\), also have one-dimensional logical spaces. Thus no distinct
enlarged sectors whose logical spaces have dimension at least two have
the same singleton
contribution. Together with
Eq.~\eqref{eq:branching-at-most-one-common}, this proves that
their contributing sets are distinct for every \(r,s\geq1\).
Exchanging the two qutrit types proves the same result for addition
of \(V^*\). At \(s=1\), the possible repeated singleton set
comes from the original highest sector \((r,1)\) and gives only
the two targets \((r,2)\) and \((r+1,0)\), both with
one-dimensional logical spaces.

\paragraph{Finding a contribution of dimension at least two.}
Let \(P\) and \(Q\) be the contributing coordinate sets of
\(\lambda\) and \(\kappa\). Suppose that every point appearing in
only one of the two sets has multiplicity one. Equality of the
enlarged logical dimensions then gives
\begin{equation*}
 0=\sum_{\delta\in P\setminus Q}M_{r,s}(\delta)
   -\sum_{\delta\in Q\setminus P}M_{r,s}(\delta)
  =|P\setminus Q|-|Q\setminus P|.
\end{equation*}
When \(s\neq2\), Lemma~\ref{lem:low-multiplicities} shows that
the only original point of multiplicity one is \((0,0)\).
The two disjoint exclusive sets must then both be empty,
contrary to the distinctness just proved.

When \(s=2\), the original points of multiplicity one are
\((0,0)\) and \((0,1)\). Distinct sets with equal numbers of
exclusive points must exchange these two points. Adding \(V\),
the only enlarged coordinate pair whose contributions include
\((0,1)\) but exclude \((0,0)\) is \((0,1)\) itself.
Its contributing set is the singleton \(\{(0,1)\}\), so its
multiplicity is one, contrary to the hypothesis.

Adding \(V^*\), an enlarged logical space of dimension at least two that
receives \((0,0)\) but not \((0,1)\) must have coordinates
\((1,1)\): the other targets of \((0,0)\) either have one-dimensional
logical spaces or also receive \((0,1)\). Its contributions include
\((1,0)\), whose logical dimension is \(M_{r,2}(1,0)=r-1\geq2\).
No target receiving \((0,1)\) can also receive \((1,0)\), by
the second branching rule in
Eq.~\eqref{eq:contributing-coordinate-sets}. This is already the
required exclusive contribution, contradicting the assumption that all
exclusive contributions have dimension one. Hence in every case a
point in only one contributing set has multiplicity at least two.
Its Dynkin label is the required \(\gamma\).
\end{proof}
\section{Independent logical control: proof and pulse construction}
\label{app:propagation-base}

The small systems illustrate the two steps used throughout the proof:
extend logical control by coupling the original contributions, then
show that gates on different sectors can be chosen independently.
We prove these starting cases analytically and use Appendix~\ref{app:sectorwise-separation} to add
further qutrits. Exact dimension calculations and a 24-pulse example
follow the completed proof as independent checks and an illustration.

\subsection{Small systems that start the induction}
\label{subsec:initial-systems}

We first treat the small systems
\((1,1)\), \((2,1)\) and \((2,2)\) separately. The first has only
one-dimensional logical subsystems, so its special-unitary target is the
identity and independent control is immediate. At \((2,1)\), the only
nontrivial logical subsystem has dimension two, where irreducibility is
sufficient to ensure independent control. At \((2,2)\), overlapping
singlet projectors also give full control on a logical subsystem of
dimension four. We then treat the starting case of the induction,
\((3,1)\). All mixed pairs used below are available by conjugating the
boundary interaction with exchanges, as in
Eq.~\eqref{eq:adjacent-and-all-pairs-generate-same-group}.

\begin{proposition}[Control on \(\mathcal H_{2,1}\)]
\label{prop:analytic-21-semiuniversality}
Independent logical control holds at \((2,1)\).
\end{proposition}

\begin{proof}
The decomposition is
\begin{equation}
 \mathcal H_{2,1}\cong U_{2,1}\otimes \mathbb{C}\oplus U_{0,2}\otimes\mathbb{C}\oplus U_{1,0}\otimes \mathbb{C}^2.
 \label{eq:analytic-21-decomposition}
\end{equation}
Only the last sector $(p=1,q=0)$ has a nontrivial logical space, of dimension two. Equation~\eqref{eq:interaction-commutant} gives irreducibility,
and Lemma~\ref{lem:two-lines-control} gives its \(\mathfrak{su}(2)\)
algebra. The global derived algebra projects onto this
\(\mathfrak{su}(2)\) algebra and vanishes on the two one-dimensional
logical blocks. It is thus the supported \(\mathfrak{su}(2)\) action
required for independent control.
Proposition~\ref{prop:global-lie-rank-criterion}
gives the finite-pulse implementation.
\end{proof}

\subsubsection{Control from overlapping singlet projectors}
\label{subsec:analytic-22-semiuniversality}

At \((2,2)\), the two nontrivial logical spaces have dimensions
four and two. A pulse generated by a rank-one projector changes the
phase of one direction. Pulses generated by projectors onto overlapping
directions generate rotations between them. A connected family spanning a logical space gives
its full unitary algebra. The different block dimensions will
then ensure that their special-unitary actions cannot be linked.

\begin{samepage}
\begin{lemma}[Projectors onto spanning vectors with connected overlaps]
\label{lem:connected-projector-control}
Let nonzero vectors \(v_1,\ldots,v_N\) span a complex Hilbert
space \(E\). Join \(j\) and \(k\) when
\(\langle v_j,v_k\rangle\neq0\). If this graph is connected,
the real Lie algebra generated by \(iP_{v_1},\ldots,iP_{v_N}\)
is \(\mathfrak u(E)\), where \(P_v\) is the orthogonal projector
onto \(\mathbb Cv\).
\end{lemma}
\end{samepage}

\begin{proof}
Order the vectors so that each after the first overlaps an earlier
one. Initially \(iP_{v_1}\) gives \(\mathfrak u(\mathbb Cv_1)\),
with zero action on its complement. Suppose the supported algebra
\(\mathfrak u(W)\) is available on the span already reached, and
write the next vector as \(v=a+b\), with \(a\in W\),
\(b\in W^\perp\). The overlap gives \(a\neq0\). If \(b=0\),
the span is unchanged. Otherwise \(A=iP_W\) and
\(B=ivv^\dagger\) are available Lie-algebra elements, and
\begin{equation*}
 [A,B]=ba^\dagger-ab^\dagger,
 \qquad [A,[A,B]]=-i(ab^\dagger+ba^\dagger).
\end{equation*}
Subtracting the cross term and the already available
\(iaa^\dagger\) from \(B\) gives \(ibb^\dagger\).
Commutators with \(\mathfrak u(W)\) supply both coupling
directions between \(b\) and each basis vector of \(W\).
Together these give the supported
\(\mathfrak u(W\oplus\mathbb Cb)\). Induction reaches \(E\).
\end{proof}

\begin{proposition}[Control on \(\mathcal H_{2,2}\)]
\label{prop:analytic-22-semiuniversality}
The four mixed interactions on \(\mathcal H_{2,2}\) alone
generate the derived logical algebra
\(\mathfrak{su}(4)\oplus\mathfrak{su}(2)\). In particular,
independent logical control holds at \((2,2)\).
\end{proposition}

\begin{proof}
The isotypic decomposition of \(\mathcal H_{2,2}\) is
\begin{equation}
\begin{aligned}
 \mathcal H_{2,2}\cong{}&U_{2,2} \otimes \mathbb{C}\oplus U_{3,0}\otimes\mathbb{C}\oplus U_{0,3}\otimes \mathbb{C}\\ &\oplus U_{1,1}\otimes \mathbb{C}^4 \oplus U_{0,0}\otimes \mathbb{C}^2.
\end{aligned}
 \label{eq:analytic-22-decomposition}
\end{equation}
Only the adjoint $(p=1,q=1)$ and singlet $(p=0,q=0)$ sectors have nontrivial logical spaces, of dimensions four and two, respectively.

For the adjoint sector, choose one mixed pair to form a singlet,
leaving the other pair to carry the adjoint representation. The
four choices define four directions in its logical space. To
compute their overlaps, realise \(U_{1,1}\) by traceless
\(3\times3\) matrices \(A\), and write the four embeddings as
\begin{equation}
\begin{aligned}
 \mathcal I_1(A)_{ij}^{k\ell}&=\delta_i^k A_j^\ell,&
 \mathcal I_2(A)_{ij}^{k\ell}&=\delta_i^\ell A_j^k,\\
 \mathcal I_3(A)_{ij}^{k\ell}&=\delta_j^k A_i^\ell,&
 \mathcal I_4(A)_{ij}^{k\ell}&=\delta_j^\ell A_i^k.
\end{aligned}
 \label{eq:analytic-22-adjoint-embeddings}
\end{equation}
The upper indices label the conjugate factors. Their inner products
are \(\langle\mathcal I_a(A),\mathcal I_b(B)\rangle
=Q_{ab}\Tr(A^\dagger B)\), with
\begin{equation}
 Q=\begin{pmatrix}
 3&1&1&0\\1&3&0&1\\1&0&3&1\\0&1&1&3
 \end{pmatrix}.
 \label{eq:analytic-22-adjoint-gram}
\end{equation}
A diagonal entry counts the three values of the contracted index.
The nonzero off-diagonal overlaps reduce to \(\Tr(A^\dagger B)\).
The other two pairs of directions have overlap
\(\overline{\Tr A}\Tr B=0\). The eigenvalues \(5,3,3,1\)
show that the four logical directions \(v_a\) are independent.

Each mixed interaction \(K_{i\bar j}\) is three times the orthogonal
projector onto the singlet of the selected pair. Within the adjoint sector, its
range is precisely the copy defined by the corresponding embedding
in Eq.~\eqref{eq:analytic-22-adjoint-embeddings}. Its logical block
is thus \(3P_{v_a}\), with respect to the logical inner product.
The nonzero overlaps in \(Q\) form a connected graph, so
Lemma~\ref{lem:connected-projector-control} gives the projected
algebra \(\mathfrak u(4)\).

For the singlet sector, the two invariant tensors
\begin{equation*}
 (\Xi_1)_{ij}^{k\ell}=\delta_i^k\delta_j^\ell,
 \qquad (\Xi_2)_{ij}^{k\ell}=\delta_i^\ell\delta_j^k
\end{equation*}
have Gram matrix
\(\left(\begin{smallmatrix}9&3\\3&9\end{smallmatrix}\right)\).
They are independent and overlap. The logical restrictions of
\(K_{1\bar1}\) and \(K_{1\bar2}\) are respectively
\(3P_{\Xi_1}\) and \(3P_{\Xi_2}\), so the same lemma gives
\(\mathfrak u(2)\).

The global derived algebra projects onto both
\(\mathfrak{su}(4)\) and \(\mathfrak{su}(2)\). By the simple-ideal
argument in Lemma~\ref{lem:propagation-semiuniversality-criterion},
these nonisomorphic factors cannot be linked. It is their direct
sum. Adjacent exchanges and the mixed interaction
at the boundary generate the four mixed interactions by
conjugation, so they give the same required derived algebra.
The overlap graph used here describes logical vectors, rather
than pairs of physical qutrits.
\end{proof}

\subsubsection{Separating the two logical qutrit sectors}

\begin{proposition}[Control on \(\mathcal H_{3,1}\)]
\label{prop:analytic-31-semiuniversality}
Independent logical control holds at \((3,1)\), with derived logical
algebra \(\mathfrak{su}(3)\oplus\mathfrak{su}(2)
\oplus\mathfrak{su}(3)\).
\end{proposition}

\begin{proof}
\emph{Control within each sector.}
Adding \(V\) to Eq.~\eqref{eq:analytic-21-decomposition} gives
the four sectors in Eq.~\eqref{eq:H31-logical-qutrit-decomposition},
with logical dimensions \(1,3,2,3\). The two qutrit spaces,
labelled \(A=(2,0)\) and \(B=(0,1)\), each receive the old
two-dimensional logical space and one old line:
\begin{equation}
\begin{aligned}
 \Gamma_{2,1}^{+}(A)&=\{(1,0),(2,1)\},\\
 \Gamma_{2,1}^{+}(B)&=\{(1,0),(0,2)\}.
\end{aligned}
 \label{eq:boundary-two-to-three-separation-obstruction}
\end{equation}
Independent control at \((2,1)\) supplies the supported \(SU(2)\)
action. The enlarged interactions act irreducibly on each logical
space because they generate its full associative matrix algebra.
Lemma~\ref{lem:two-one-control} then gives the full \(SU(3)\)
action on each of \(A\) and \(B\).

The remaining nontrivial logical space, labelled \((1,2)\),
has dimension two. Irreducibility and
Lemma~\ref{lem:two-lines-control} give its full \(SU(2)\) action.

\emph{Independence of the two qutrit sectors.}
The common contribution in Eq.~\eqref{eq:boundary-two-to-three-separation-obstruction}
has dimension two, and both exclusive contributions are lines.
Thus its old determinant-one gates do not distinguish the two sectors.
Instead compare the exchange \(S=S_{12}\) with \(K=K_{3\bar1}\).
Every trace below is on the relevant three-dimensional logical space.

The decompositions
\begin{equation*}
\begin{aligned}
 \operatorname{Sym}^2V\otimes V^*&\cong U_{2,1}\oplus U_{1,0},\\
 \bigwedge^2V\otimes V^*&\cong U_{0,2}\oplus U_{1,0}
\end{aligned}
\end{equation*}
show that the original fundamental multiplicity space has one
symmetric and one antisymmetric direction under \(S_{12}\). The exchange has trace zero
there, and acts by \(+1\) on the extra line \((2,1)\) or by
\(-1\) on \((0,2)\). Thus \(\Tr S_A=1\) and \(\Tr S_B=-1\).
The range of \(K\) is
\(V_1\otimes V_2\otimes\mathbb C\Omega_{3\bar1}\).
Its symmetric and antisymmetric parts each contain one copy of the
corresponding representation \(A\) or \(B\). Hence \(K\) has
logical rank one and trace three in each sector, with \(S\) acting
as \(+1\) or \(-1\) on that range. Thus
\begin{equation*}
\begin{array}{c|rrr}
 &\Tr S&\Tr K&\Tr(SK)\\ \hline
 A&1&3&3\\
 B&-1&3&-3
\end{array}
\end{equation*}
For \(\widetilde H_\lambda
=H_\lambda-(\Tr H_\lambda)I_3/3\), these values give
\begin{equation}
 \Tr(\widetilde S_A\widetilde K_A)=2,
 \qquad \Tr(\widetilde S_B\widetilde K_B)=-2.
 \label{eq:initial-qutrit-trace-separation}
\end{equation}
If one simple Lie algebra acted on both qutrit sectors, the
identification between its actions would preserve the trace inner
product on traceless generators. The opposite signs rule out that
identification. To see this precisely, write
the compact full logical Lie algebra as
\(\mathfrak L=\mathfrak z(\mathfrak L)\oplus\mathfrak D\),
where \(\mathfrak D=[\mathfrak L,\mathfrak L]\).
Its centre is scalar on each irreducible logical block, and every
block of \(\mathfrak D\) is traceless. Consequently the derived
component of \(iH\) projects to \(i\widetilde H_\lambda\).
If the two \(\mathfrak{su}(3)\) actions were linked, the
simple-ideal argument of Lemma~\ref{lem:propagation-semiuniversality-criterion}
would identify them by one Lie-algebra automorphism. Such an
automorphism preserves the trace inner product \(-\Tr(XY)\)
on \(\mathfrak{su}(3)\), since the Killing form is \(6\Tr(XY)\).
Taking the two derived components of \(iS\) and \(iK\) contradicts
Eq.~\eqref{eq:initial-qutrit-trace-separation}. The
\(\mathfrak{su}(2)\) factor cannot be linked to either
\(\mathfrak{su}(3)\), since they are nonisomorphic. This proves
the independent direct sum, and Lemma~\ref{lem:pulse-group-closed}
and Proposition~\ref{prop:global-lie-rank-criterion} give its exact
finite-pulse implementation.

Subtracting traces here only identifies derived-algebra components.
The physical Hamiltonians and their scalar phases are unchanged.
\end{proof}

\subsection{The general induction step}
\label{sec:mixed-propagation-theorem}
\label{subsec:propagation-theorem}

We start the induction at \((3,1)\) and add either type of qutrit.
For either addition, retain the notation
\(W_\gamma,W'_\lambda,\Gamma(\lambda)\) from
Appendix~\ref{app:branching-multiplicity}.

\begin{theorem}[Induction after adding one qutrit]
\label{thm:mixed-propagation}
Let \(r\geq3\), \(s\geq1\), and assume
\(G_{r,s}\supseteq\mathcal U_{\mathrm{log}}(r,s)\). Then
\begin{equation}
 G_{r+1,s}\supseteq\mathcal U_{\mathrm{log}}(r+1,s),
 \qquad G_{r,s+1}\supseteq\mathcal U_{\mathrm{log}}(r,s+1).
 \label{eq:mixed-induction-conclusion}
\end{equation}
\end{theorem}

\begin{proof}
Fix either addition. A one-dimensional logical space has trivial \(SU(1)\) action.
For an enlarged logical space of dimension greater than one, a contribution
of dimension at least three gives full control by Lemma~\ref{lem:regular-sector-controllability}.
Otherwise, the possible lists of predecessor dimensions are \((2)\),
\((2,1)\), and \((1,1)\), covered by
Lemma~\ref{lem:thin-sector-controllability}.

For any two distinct enlarged sectors whose logical spaces have the
same dimension, at least two, Lemma~\ref{lem:branching-separation}
supplies a predecessor of dimension at least two present in exactly one
of them. Independent control in the original system gives a nonscalar
gate on that predecessor and the identity on every other original
logical subsystem. Its embedding distinguishes the enlarged sectors by
Eq.~\eqref{eq:branching-existing-action}.
Lemma~\ref{lem:propagation-semiuniversality-criterion}, applied to the
compact connected enlarged group, gives independent control in the
enlarged system.
\end{proof}

The first applications give independent control at \((4,1)\) and
\((3,2)\) by adding one fundamental or one dual qutrit to \((3,1)\),
respectively. Table~\ref{tab:initial-five-qutrit-branching} lists the
dimensions of the predecessors of each enlarged sector. A contribution of dimension three extends
control to its enlarged space. The remaining nontrivial cases are a two-dimensional contribution alone or coupled to a line, as covered in Appendix~\ref{app:sectorwise-separation}.

\begin{table}[tbp]
\caption{Adding either type of qutrit to \((3,1)\).
Each sum lists the dimensions of the original logical spaces
contributing to the enlarged sector. Its value is the new logical
dimension.}
\label{tab:initial-five-qutrit-branching}
\begin{ruledtabular}
\begin{tabular}{c c c}
Added type & Enlarged sector & Contributions\\ \hline
\(V\)&\((4,1)\)&\(1\)\\
&\((3,0)\)&\(1+3\)\\
&\((2,2)\)&\(1+2\)\\
&\((1,1)\)&\(3+2+3\)\\
&\((0,3)\)&\(2\)\\
&\((0,0)\)&\(3\)\\ \hline
\(V^*\)&\((3,2)\)&\(1\)\\
&\((4,0)\)&\(1\)\\
&\((2,1)\)&\(1+3+2\)\\
&\((1,0)\)&\(3+3\)\\
&\((1,3)\)&\(2\)\\
&\((0,2)\)&\(2+3\)
\end{tabular}
\end{ruledtabular}
\end{table}

\begin{corollary}[Equivalence under complex conjugation]
\label{cor:conjugation-dual-semiuniversality}
For \(r,s\geq1\), independent logical control on
\(\mathcal H_{r,s}\) is equivalent to independent logical control
on \(\mathcal H_{s,r}\).
\end{corollary}

\begin{proof}
Coordinatewise complex conjugation followed by tensor-factor
reordering interchanges the two types and maps sector \((p,q)\)
to \((q,p)\) with the same multiplicity. It maps exchanges to
exchanges and \(K_{i\bar j}\) to \(K_{j\bar i}\). For this
antiunitary map \(\mathsf C\),
\begin{equation*}
 \mathsf C e^{-itH}\mathsf C^{-1}
 =e^{+it\mathsf C H\mathsf C^{-1}}.
\end{equation*}
Both signs of \(t\) are allowed, so the physical pulse groups
and their independent logical targets are identified.
\end{proof}

\subsection{Proof of Theorem~\ref{thm:global-semiuniversality}}

Proposition~\ref{prop:analytic-31-semiuniversality} gives independent
control at \((3,1)\). Repeated applications of
Theorem~\ref{thm:mixed-propagation} cover every
\(r\geq3,s\geq1\).
Corollary~\ref{cor:conjugation-dual-semiuniversality} then covers
\(r\geq1,s\geq3\).

The remaining systems are \((1,1),(2,1),(1,2),(2,2)\).
The target at \((1,1)\) is trivial.
Proposition~\ref{prop:analytic-21-semiuniversality} and conjugation
give \((2,1)\) and \((1,2)\), while
Proposition~\ref{prop:analytic-22-semiuniversality} gives \((2,2)\).
This proves the theorem for all \(r,s\geq1\).

\subsection{Independent exact checks of the Lie-algebra dimensions}
\label{subsec:exact-rank-procedure}
\label{app:exact-base-calculation}
\label{subsec:small-system-calculations}

\begin{table}[tbp]
\caption{Independent exact dimension checks. Multiplicities include the
one-dimensional logical spaces. The full ranks are computed over
\(\mathbb F_{1\,000\,003}\). Each reaches the analytic upper bound
\(d+c\), which also fixes the derived dimension to \(d\).}
\label{tab:small-system-ranks}
\begin{ruledtabular}
\begin{tabular}{c l r r r}
 \((r,s)\)&Multiplicities&\(d\)&\(c\)&Full rank\\
 \hline
 \((1,1)\)&\(1,1\)&0&1&1\\
 \((2,1)\)&\(1,1,2\)&3&2&5\\
 \((3,1)\)&\(1,2,3,3\)&19&2&21\\
 \((4,1)\)&\(1,2,3,3,4,8\)&97&2&99\\
 \((3,2)\)&\(1,1,2,5,6,6\)&97&3&100
\end{tabular}
\end{ruledtabular}
\end{table}

The preceding proof uses no computed rank. The following exact
calculations independently check the dimensions of the full Lie
algebras generated by the physical interactions and of their scalar parts.

For a fixed system, write
\(\mathfrak L=\Lie_{\mathbb R}\{iH_a\}_a\) and
\(\mathfrak D=[\mathfrak L,\mathfrak L]\), using the interactions
on all pairs. The sign change from \(-iH_a\) to \(iH_a\) does
not change a real Lie algebra. These interactions generate the same
group and Lie algebra as the adjacent interactions by
Eq.~\eqref{eq:adjacent-and-all-pairs-generate-same-group}.
Set \(d=\sum_\lambda(m_\lambda^2-1)\), and let \(c\) be the
number of interaction types present: exchanges within each type of
qutrit, and mixed interactions. The dimension test in Appendix~\ref{app:mathematical-framework}
gives \(\dim\mathfrak D\leq d\).

Block traces give the other dimension bound. Define
\(\tau(X)=(\Tr X_\lambda)_\lambda\).
Compactness gives
\(\mathfrak L=\mathfrak z(\mathfrak L)\oplus\mathfrak D\).
The logical actions are irreducible by
Eq.~\eqref{eq:interaction-commutant}, so the centre is scalar on
each block, while every block of \(\mathfrak D\) is traceless.
Thus an element of \(\mathfrak L\) has zero block traces exactly
when it belongs to \(\mathfrak D\).
Every nested commutator has zero block traces, so
\(\tau(\mathfrak L)\) is spanned by the trace vectors of the
original generators \(iH_a\). Interactions of the same type are
related by permutations within each qutrit type and have identical
block traces. There are \(c\) types, giving
\begin{equation*}
 \dim(\mathfrak L/\mathfrak D)=\dim\tau(\mathfrak L)\leq c.
\end{equation*}

It follows that a full-algebra rank reaching \(d+c\) determines
both dimensions:
\begin{equation}
 \dim\mathfrak L=d+c
 \quad\Longrightarrow\quad
 \dim\mathfrak D=d.
 \label{eq:full-rank-determines-derived}
\end{equation}
In fact \(\dim\mathfrak L\leq d+c\), while equality forces both
upper bounds to be reached. The extra \(c\) dimensions account for
the possible scalar contributions. They do not imply arbitrary
relative phases between all sectors.

To obtain the full-algebra lower bound, construct the physical
integer matrices in the computational basis
\(\ket{\mathbf a;\bar{\mathbf b}}\), where
\(a_i,b_j\in\{0,1,2\}\). An exchange permutes two entries, and
\begin{equation}
 K_{i\bar j}\ket{\mathbf a;\bar{\mathbf b}}
 =\delta_{a_i,b_j}\sum_{u=0}^2
 \ket{\mathbf a[i\leftarrow u];
       \overline{\mathbf b[j\leftarrow u]}}.
 \label{eq:small-system-contraction-action}
\end{equation}
Represent a skew-Hermitian matrix by its real antisymmetric part
\(R\) and imaginary symmetric part \(J\). The identity
\begin{equation*}
 [iH_a,R+iJ]=-[H_a,J]+i[H_a,R]
\end{equation*}
keeps every nested commutator integral. Reduce these pairs over
\(\mathbb F_\wp\), with \(\wp=1\,000\,003\), and repeatedly
bracket the generators with each newly independent vector. Stop
when no new independent vectors remain.

The resulting rank is a lower bound on the real dimension: a minor
with nonzero determinant over \(\mathbb F_\wp\) has nonzero
integer determinant. When this lower bound equals \(d+c\), the
two real dimensions follow from
Eq.~\eqref{eq:full-rank-determines-derived}. A smaller rank over a
finite field would not prove a smaller real dimension.

The calculation retains the entries on and above the diagonal
between computational basis states of the same collective weight.
For such a state, the weight counts are
\(\nu_u=\#\{i:a_i=u\}-\#\{j:b_j=u\}\).
All interactions and commutators preserve these counts. Symmetry
or antisymmetry determines the lower triangle, so the retained
coordinates contain the full matrix information.
% AUTHOR TODO: verify_initial_cases.py is not present in the supplied project.
% Supply the matching exact-check script, iteration histories, coordinate
% counts and reproduction instructions before citing an accompanying archive.

\begin{proposition}[Exact dimension checks]
\label{prop:small-mixed-system-calculations}
\label{prop:exact-32-base-calculation}
For each system in Table~\ref{tab:small-system-ranks}, the full
finite-field rank is \(d+c\). The corresponding real dimensions
are \(\dim\mathfrak L_{r,s}=d+c\) and
\(\dim\mathfrak D_{r,s}=d\), and independent logical control holds:
\begin{equation}
\begin{aligned}
 \mathfrak D_{r,s}=\bigoplus_{\lambda\in\Lambda_{r,s}}
 \mathfrak{su}(W_{r,s}^\lambda),\\
 G_{r,s}\supseteq\mathcal U_{\mathrm{log}}(r,s).
\end{aligned}
 \label{eq:small-system-semiuniversality}
\end{equation}
\end{proposition}

\begin{proof}
The tensor-product rules give the multiplicities in the table.
The exact calculation just described gives its final column.
The matching upper bounds and
Eq.~\eqref{eq:full-rank-determines-derived} prove the dimension
equalities. Since \(\mathfrak D_{r,s}\) is contained in the
displayed direct sum of the same dimension, the two algebras are
equal. Proposition~\ref{prop:global-lie-rank-criterion} gives the
group conclusion.

For \((3,1)\), the derived dimension \(19=8+3+8\)
independently confirms the two separate \(\mathfrak{su}(3)\)
factors proved in Proposition~\ref{prop:analytic-31-semiuniversality}.
\end{proof}

\subsection{A 24-pulse shift on one logical qutrit}
\label{subsec:explicit-pulses-31}

\begin{table*}[!t]
\caption{Pulse parameters for applying a qutrit shift to \(W_{3,1}^{2,0}\)
while leaving the other logical subsystems in \(\mathcal H_{3,1}\) unchanged.
The entry \(t_{k,j}\) specifies the gate \(e^{-it_{k,j}\Psi_j}\).
Read rows from top to bottom and apply pulses generated by
\(S_{12}\), \(S_{23}\), and \(K_{3\bar1}\), in that order within each row.
All displayed digits are used in the numerical checks. The target and
product conventions are Eqs.~\eqref{eq:selective-qutrit-target} and
\eqref{eq:explicit-pulse-order-31}.}
\label{tab:explicit-pulse-areas-31}
\begin{ruledtabular}
\begin{tabular}{c r r r}
Cycle \(k\)&\(t_{k,1}\) (\(S_{12}\))&\(t_{k,2}\) (\(S_{23}\))&\(t_{k,3}\) (\(K_{3\bar1}\))\\
\hline
1& 0.89890265074673437& 1.6863543126948368& 0.78807736635071035\\
2&-0.81955596050484836& 1.5896057342665548&-0.9505817084235455\\
3& 1.7513191311346539&-1.8540482119560606& 0.23672254751977384\\
4& 2.0380528681525618&-2.2039994317951597&-0.65163453775604407\\
5&-0.54136253487897834&-1.6273143278945019& 0.36735616246176561\\
6& 1.9241651693991457&-0.82171433305790531& 0.67814693186966712\\
7&-2.8509494846340813& 1.6419160964371304& 0.62804715794838684\\
8& 2.8552348936408816& 2.6165787354286216& 0.99826118242248163
\end{tabular}
\end{ruledtabular}
\end{table*}

\begin{revision}
Symmetry-invariant gates act as the identity on each gauge factor.
For \((r,s)=(3,1)\), the logical spaces have dimensions \(1,3,2,3\).
The 81-dimensional physical calculation thus reduces to block-diagonal
matrices acting on a nine-dimensional space. For any two invariant
unitaries \(U,T\),
\begin{equation}
 \|U-T\|_{\mathrm{op}}
 =\max_\lambda\|U_\lambda-T_\lambda\|_{\mathrm{op}}.
 \label{eq:pulse-block-norm}
\end{equation}
Each copy of an irreducible representation has one highest-weight
line. The highest-weight space of a sector thus has exactly its
logical dimension and carries the logical action directly. Writing
\(Q_\lambda\) for an orthonormal embedding of that space gives the
reduced controls below. The basis convention after the display fixes
both \(Q_\lambda\) and the target shift.

Set \(H_1=S_{12}\), \(H_2=S_{23}\), \(H_3=K_{3\bar1}\), and
\(\Psi_j^\lambda=Q_\lambda^\dagger H_jQ_\lambda\), with
\(\Psi_j=\bigoplus_\lambda\Psi_j^\lambda\). Define
\begin{equation}
 \mathsf M_c=\begin{pmatrix}
 c&\sqrt{1-c^2}\\ \sqrt{1-c^2}&-c
 \end{pmatrix},\qquad P_0=\operatorname{diag}(1,0,0).
 \label{eq:pulse-reflection-block}
\end{equation}
The complete reduced controls are
\begin{equation}
\begin{array}{c|ccc}
 \lambda&\Psi_1^\lambda&\Psi_2^\lambda&\Psi_3^\lambda\\ \hline
 (3,1)&1&1&0\\
 (2,0)&1\oplus\mathsf M_{1/4}&\mathsf M_{1/3}\oplus1&3P_0\\
 (1,2)&\mathsf M_{1/2}&\operatorname{diag}(-1,1)&0_2\\
 (0,1)&(-1)\oplus\mathsf M_{1/2}&\mathsf M_{1/3}\oplus(-1)&3P_0
\end{array}
 \label{eq:explicit-reduced-controls-31}
\end{equation}
These matrices satisfy \(\Psi_1^2=\Psi_2^2=I_9\) and
\(\Psi_3^2=3\Psi_3\). Their block traces are
\((1,1,0,-1)\) for either exchange and \((0,3,0,3)\) for the
singlet interaction. These scalar contributions must be retained
in the physical pulses. The dimensions \(21\) and \(19\) of the full
and derived Lie algebras, respectively, are independently checked in
Table~\ref{tab:small-system-ranks}.
% Reproduction: section_III_C_reproducibility/reduced_model.py reconstructs
% the exact physical interactions and highest-weight reductions.
% verify_pulses.py independently checks their physical realisation.

\paragraph{Basis convention for reproducing the target.}
The logical basis fixes the shift \(X_3\) in
Eq.~\eqref{eq:selective-qutrit-target}. Order the physical orthonormal
basis lexicographically as \(\ket{a b c\bar d}\), with
\(a,b,c,d\in\{0,1,2\}\).  For \(E_{ij}=\ket i\!\bra j\), the collective
infinitesimal action is
\begin{equation}
 \mathcal J_{ij}=E_{ij}^{(1)}+E_{ij}^{(2)}+E_{ij}^{(3)}
                   -E_{ji}^{(\bar1)}.
 \label{eq:pulse-collective-generators}
\end{equation}
The Cartan operators are \(\mathcal J_{00}-\mathcal J_{11}\) and
\(\mathcal J_{11}-\mathcal J_{22}\). For each sector label \((p,q)\), select
the computational basis vectors of that weight in lexicographic order
and compute the rational nullspace of the vertically stacked restrictions
of \(\mathcal J_{01}\) and \(\mathcal J_{12}\).  Take its standard
free-variable basis in ascending column order, apply Gram--Schmidt in
that order, make the first nonzero coefficient of each vector positive,
and normalise with the positive square root of its squared norm.
Writing \(m_\lambda=\dim W_{3,1}^{\lambda}\), this specifies orthonormal
highest-weight embeddings
\(Q_\lambda:\mathbb C^{m_\lambda}\to\mathcal H_{3,1}\) and fixes the
multiplicity basis, including the basis of \(X_3\).

Each cycle applies pulses generated by \(S_{12}\), \(S_{23}\), and
\(K_{3\bar1}\), in that order.
The numbers \(t_{k,1},t_{k,2},t_{k,3}\) determine the corresponding gates
\(e^{-it_{k,j}\Psi_j}\).  The complete sequence is
\begin{equation}
\begin{aligned}
 U_{\rm pulse}&=C_8C_7\cdots C_1,\\
 C_k&=e^{-it_{k,3}\Psi_3}e^{-it_{k,2}\Psi_2}e^{-it_{k,1}\Psi_1}.
\end{aligned}
 \label{eq:explicit-pulse-order-31}
\end{equation}

Table~\ref{tab:explicit-pulse-areas-31} gives all 24 pulse parameters, using
the same precision as the exported decimal data. Negative values may be replaced
by adding \(2\pi\) for an exchange or \(2\pi/3\) for a singlet interaction.
These are exact periods of the corresponding gates.  For a constant Hamiltonian
\(\hbar g_j H_j\) applied for a duration \(\Delta t\), the
parameter is \(t=g_j\Delta t\).  Converting the tabulated values to pulse
durations thus requires the coupling strengths \(g_j\).

The sequence was obtained by nonlinear least squares on the real and
imaginary entries of \(U_{\rm pulse}-T_X\), using analytic derivatives
and no fitting of separate phases on the blocks. The reduced product
can be checked directly from the displayed data using
\begin{equation*}
\begin{aligned}
 e^{-it\Psi_j}&=\cos(t)I-i\sin(t)\Psi_j,\qquad j=1,2,\\
 e^{-it\Psi_3}&=I+(e^{-3it}-1)\Psi_3/3.
\end{aligned}
\end{equation*}
A 90-digit evaluation using all displayed pulse parameters gives a
maximum block error of approximately \(2.6\times10^{-15}\),
without adjusting any phase. Equation~\eqref{eq:pulse-block-norm} gives
the same error on the physical space. No relative phases between sectors
are discarded. This is a floating-point evaluation, not an analytic
identity or an interval bound.

% The section_III_C_reproducibility package retains the original optimisation
% record in pulses_31.json. It includes the full 81-dimensional verifier and
% saved results of a fresh optimisation with independent verification.
\end{revision}

\section{Teleportation under correlated noise: channels, thresholds and resources}
\label{app:teleportation-analysis}

In this appendix, we derive the effective qutrit teleportation channels
under the mixture of collective and independent noise. We determine
when the encoded protocol has higher fidelity than the unencoded
protocol and when it preserves entanglement while the unencoded
channel is entanglement breaking. We then separate fidelity on retained
trials from the probability of obtaining an output, explain how to test
the noise model, and evaluate the physical costs of the encoding.
We finish by considering other collective processes and additional
logical noise.

\subsection{Logical channels after noise and recovery}
\label{subsec:finite-weyl-reduction}

Encoding in the multiplicity subsystem protects the logical state from
the collective errors considered here, which act only on the gauge
subsystem. Recovery discards the gauge subsystem and returns the logical
state. We define these maps for a chosen sector \(U\otimes W\), writing
\(d=\dim W\geq2\). In the main example,
\(U=U_{2,0}\), \(W=W_{3,1}^{2,0}\), and \(d=3\).
Write \(\mathcal H=\mathcal H_{r,s}\), and let
\(J:U\otimes W\to\mathcal H\) be the block isometry,
\(P=JJ^\dagger\), and \(\tau\) a fixed gauge state. Encoding is defined by
\begin{equation}
    \operatorname{Enc}(\rho)
    :=J(\tau\otimes\rho)J^\dagger.
\end{equation}
For trials with no signalled loss, we use the recovery map
\begin{equation}
 \operatorname{Rec}(A)
 :=\Tr_U(J^\dagger AJ)
 +\Tr[(I-P)A]\frac{I_W}{d}.
 \label{eq:recovery-master}
\end{equation}
The first term extracts the logical component within the selected
sector. The second replaces the component outside that sector by the
maximally mixed logical state. Both terms define completely positive
maps. Their traces are \(\Tr(PA)\) and \(\Tr[(I-P)A]\), whose sum is
\(\Tr A\), so recovery is CPTP. These definitions apply to any encoding
sector.

\begin{samepage}
\begin{lemma}[Collective noise acts only on the gauge subsystem]
\label{lem:collective-factorization-teleportation}
Suppose a channel \(\mathcal M\) preserves the selected sector and
its Kraus operators restrict to \(B_a\otimes I_W\), with
\(\sum_a B_a^\dagger B_a=I_U\). Then
\(\operatorname{Rec}\circ\mathcal M\circ\operatorname{Enc}=\id_W\).
\end{lemma}
\end{samepage}

\begin{proof}
For every logical input \(\rho\), recovery gives
\begin{equation*}
\begin{aligned}
 (\operatorname{Rec}\circ\mathcal M\circ\operatorname{Enc})(\rho)
 &=\sum_a\Tr(B_a\tau B_a^\dagger)\rho\\
 &=\Tr\!\left(\tau\sum_a B_a^\dagger B_a\right)\rho=\rho.
\end{aligned}
\end{equation*}
\end{proof}

For the qutrit Weyl operators \(D_g=X^aZ^b\), with
\(g=(a,b)\in\mathbb Z_3^2\),
\begin{equation*}
 \frac19\sum_gD_g A D_g^\dagger=\Tr(A)I_3/3.
\end{equation*}
Averaging over the powers of \(Z\) removes off-diagonal entries, and
averaging over the powers of \(X\) equalises the diagonal entries.
Independent averaging on the four physical factors, with conjugate
operators on \(V^*\), gives
\(\mathcal T_{\rm ind}(A)=\Tr(A)I_{81}/81\).
Recovery sends \(I_{81}/81\) to \(I_3/3\): both terms in
Eq.~\eqref{eq:recovery-master} are proportional to \(I_3\), and their
traces sum to one.

In the collective branch \(\mathcal T_{\rm col}\), the same label
\(g\) is used on all four factors. Since
\(J^\dagger R_{3,1}(D_g)J=\varrho_{2,0}(D_g)\otimes I_W\),
Lemma~\ref{lem:collective-factorization-teleportation} applies.
Composing encoding, either noise channel and recovery gives
\begin{equation}
\begin{aligned}
 \operatorname{Rec}\circ\mathcal T_{\rm col}\circ\operatorname{Enc}
 &=\id_3,\\
 \operatorname{Rec}\circ\mathcal T_{\rm ind}\circ\operatorname{Enc}
 &=\Delta_3,
\end{aligned}
 \label{eq:weyl-recovered-branches}
\end{equation}
where \(\Delta_3(A):=\Tr(A)I_3/3\).
Table~\ref{tab:weyl-branch-actions} collects the three branches of
Eq.~\eqref{eq:weyl-noise-channel}. Their weighted sum gives the depolarising channel
\(\mathcal D_\varepsilon^{(3)}\) for the unencoded protocol and
the depolarising channel \(\mathcal D_{(1-\mu)\varepsilon}^{(3)}\) for the encoded protocol.

So far we have averaged the physical errors and applied recovery to
obtain a channel on Bob's logical share of the resource. Standard
teleportation introduces one further channel average. If a channel
\(\mathcal M\) acts on Bob's half of a maximally entangled state of
local dimension \(d\), ideal standard teleportation induces the channel
\begin{equation}
 \mathsf T_d(\mathcal M)
 =\frac1{d^2}\sum_{g\in\mathbb Z_d^2}
 \operatorname{Ad}_{D_g^{(d)}}^{-1}\circ\mathcal M
 \circ\operatorname{Ad}_{D_g^{(d)}},
 \label{eq:resource-to-teleportation-twirl}
\end{equation}
up to the Bell outcome convention~\cite{BowenBose2001}.
Here \(D_{(a,b)}^{(d)}=X_d^aZ_d^b\) is a qudit Weyl operator, where
\(X_d\ket j=\ket{j+1\bmod d}\) and
\(Z_d\ket j=e^{i2\pi j/d}\ket j\), and
\(\operatorname{Ad}_D(A):=DAD^\dagger\).
The map \(\mathsf T_d\) averages channels under conjugation, whereas
\(\mathcal T_{\rm col}\) and \(\mathcal T_{\rm ind}\) average
physical errors acting on states. Depolarising channels are fixed by
\(\mathsf T_d\), so the two depolarising transmission channels above
are also the effective channels of ideal standard teleportation.

For the logical qutrit, Theorem~\ref{thm:global-semiuniversality}
supplies a logical Weyl
correction \(\widetilde D_g\) acting as \(I_U\otimes D_g^{\rm L}\) on the selected block and as the identity elsewhere. Recovery obeys
\begin{equation}
 \operatorname{Rec}(\widetilde D_g A\widetilde D_g^\dagger)
 =D_g^{\rm L}\operatorname{Rec}(A)(D_g^{\rm L})^\dagger.
 \label{eq:recovery-displacement-covariance}
\end{equation}
The replacement state \(I_3/3\) is invariant under this correction.
Thus Bob can apply the Bell correction before recovery.

\subsection{Fidelity thresholds and entanglement transmission}
\label{subsec:weyl-threshold-derivation}

\subsubsection{One fidelity formula for both protocols}

Assume the conditions on retained trials and channel composition in
Sec.~\ref{sec:baseline-calibration}. In particular, a fixed CPTP
baseline channel \(\mathcal E_\alpha^{(0)}\) follows the ideal
effective channel, for \(\alpha\in\{{\rm u},{\rm enc}\}\).
In dimension \(d\), write the depolarising map and average fidelity as
\begin{equation*}
\begin{aligned}
 \mathcal D_\xi^{(d)}(A)&=(1-\xi)A+\xi\Tr(A)I_d/d,\\
 F_{\rm avg}(\mathcal M)&=\int d\psi\,
 \langle\psi|\mathcal M(\ketbra{\psi})|\psi\rangle.
\end{aligned}
\end{equation*}
The integral uses normalised Haar measure on pure states, and the Weyl
mixture has \(\xi\in[0,1]\). We define the rescaled fidelities
\begin{equation}
 a_\alpha:=\frac{dF_\alpha-1}{d-1},\qquad
 v_\alpha:=a_\alpha(0)=\frac{dF_\alpha^{(0)}-1}{d-1}.
 \label{eq:rescaled-fidelity-coordinate}
\end{equation}
These rescaled fidelities can be negative for general channels.
Perfect transmission has \(a=1\), and any channel with a fixed
output has \(a=0\). Linearity and the Haar average of
a pure input give
\begin{equation}
 F_{\rm avg}(\mathcal E_\alpha^{(0)}\circ\mathcal D_\xi^{(d)})
 =(1-\xi)F_\alpha^{(0)}+\frac{\xi}{d}.
 \label{eq:affine-baseline-fidelity}
\end{equation}
Indeed, the completely mixed branch has the fixed output
\(\mathcal E_\alpha^{(0)}(I_d/d)\), whose average overlap with
uniformly sampled pure inputs is \(1/d\). No unitality or
isotropy assumption on the baseline channel is needed.

The two rescaled fidelities become
\begin{equation}
\begin{aligned}
 a_{\rm u}&=v_{\rm u}(1-\varepsilon),\\
 a_{\rm enc}&=v_{\rm enc}[1-(1-\mu)\varepsilon].
\end{aligned}
 \label{eq:general-weyl-fidelities}
\end{equation}
For the qutrit protocols, we have
\(F_\alpha=1/3+2a_\alpha/3\). The classical fidelity bound
\(2/(d+1)\)~\cite{Horodecki1999} becomes the horizontal level \(a=1/(d+1)\),
which is \(1/4\) for qutrits. The fidelity crossover is obtained by
intersecting these two lines. Identifying a channel requires more
information than its fidelity curve, as the example in
Appendix~\ref{subsec:experimental-calibration-protocol} shows.

\subsubsection{Fidelity crossover and classical thresholds}

The crossover is the intersection of the two curves in
Eq.~\eqref{eq:general-weyl-fidelities}. Their intersections with
\(a=1/4\) give the qutrit classical thresholds. The superscript
\(\mathrm{MP}\) refers to measure-and-prepare protocols.
These thresholds need not coincide with the entanglement-breaking (EB) boundary of the complete channel.

\begin{samepage}
\begin{proposition}[Crossover and classical fidelity thresholds]
\label{prop:fidelity-classical-benchmark-master}
Suppose the fidelities on retained trials obey
Eq.~\eqref{eq:general-weyl-fidelities}, with
\(1/4<v_{\rm enc}\leq v_{\rm u}\leq1\) and
\(0\leq\mu\leq1\). Except when \(v_{\rm enc}=v_{\rm u}\)
and \(\mu=0\), the encoded fidelity is larger exactly when
\(\varepsilon>\varepsilon_\times\), with thresholds in
Table~\ref{tab:teleportation-thresholds}.
An interval in which \(F_{\rm u}\leq1/2<F_{\rm enc}\) exists
exactly when
\begin{equation}
 v_{\rm enc}[1+\mu(4v_{\rm u}-1)]>v_{\rm u}.
 \label{eq:classical-benchmark-condition-master}
\end{equation}
For \(\mu<1\), that interval is
\([\varepsilon_{\rm u}^{\rm MP},\varepsilon_{\rm enc}^{\rm MP})
\cap[0,1]\). For \(\mu=1\), it is
\([\varepsilon_{\rm u}^{\rm MP},1]\).
\end{proposition}
\end{samepage}

\begin{table}[tbp]
\caption{Qutrit thresholds under
\(1/4<v_{\rm enc}\leq v_{\rm u}\leq1\). The crossover formula
excludes identical curves. The encoded classical threshold assumes
\(\mu<1\). All intervals are restricted to \([0,1]\).}
\label{tab:teleportation-thresholds}
\begin{ruledtabular}
\begin{tabular}{l c}
Comparison & Threshold\\ \hline
\(F_{\rm enc}=F_{\rm u}\) &
\(\displaystyle\varepsilon_\times=
\frac{v_{\rm u}-v_{\rm enc}}{v_{\rm u}-(1-\mu)v_{\rm enc}}\)\\[6pt]
\(F_{\rm u}=1/2\) &
\(\displaystyle\varepsilon_{\rm u}^{\rm MP}=1-\frac1{4v_{\rm u}}\)\\[6pt]
\(F_{\rm enc}=1/2\) &
\(\displaystyle\varepsilon_{\rm enc}^{\rm MP}
=\frac{1-(4v_{\rm enc})^{-1}}{1-\mu}\)\\[6pt]
Unencoded depolarising boundary & \(\varepsilon_{\rm dep}^{\rm EB}=3/4\)
\end{tabular}
\end{ruledtabular}
\end{table}

\begin{proof}
Subtracting the rescaled fidelities gives
\begin{equation*}
 a_{\rm enc}-a_{\rm u}
 =v_{\rm enc}-v_{\rm u}
 +\varepsilon[v_{\rm u}-(1-\mu)v_{\rm enc}].
\end{equation*}
Its slope is positive except when the curves coincide. Setting
this difference to zero, or either line to \(1/4\), gives the
values of Table~\ref{tab:teleportation-thresholds}. An interval where only encoding exceeds the classical bound exists exactly when its line is still above that level as the
unencoded line reaches it:
\begin{equation*}
 a_{\rm enc}(\varepsilon_{\rm u}^{\rm MP})>\frac14.
\end{equation*}
Substitution gives Eq.~\eqref{eq:classical-benchmark-condition-master}.
\end{proof}

The complete channel model in Sec.~\ref{sec:baseline-calibration}
is one sufficient route to these curves. The branch probabilities
and fidelities in Appendix~\ref{subsec:branch-resolved-data} also
give their scalar predictions without establishing that channel
model. The proposition uses only the curves, not a conclusion
about when either channel becomes entanglement breaking.

Solving the same condition for the collective weight gives
\begin{equation*}
 \mu>\mu_{\rm crit}
 :=\frac{v_{\rm u}-v_{\rm enc}}{v_{\rm enc}(4v_{\rm u}-1)}.
\end{equation*}
If \(v_{\rm enc}=v_{\rm u}>0\) and \(\mu>0\), encoding improves
fidelity for every \(\varepsilon>0\). If \(\mu=0\), an encoding
with \(v_{\rm enc}\leq v_{\rm u}\) has no fidelity advantage.
The curves coincide when both equalities hold. If
\(v_{\rm enc}>v_{\rm u}\), the encoded protocol already starts
with higher fidelity. A protocol with \(v_\alpha\leq1/4\)
never exceeds the classical bound under these curves.
When \(\varepsilon_{\rm enc}^{\rm MP}>1\), its fidelity stays
above the bound throughout the physical interval, including
\(\varepsilon=1\). If \(\varepsilon_{\rm enc}^{\rm MP}=1\),
equality holds at that endpoint, so it is excluded from the
interval of strict advantage.

The same algebra applies in logical dimension \(d\) whenever the
two curves have the form in Eq.~\eqref{eq:general-weyl-fidelities}.
In the classical-threshold formulas and the expression for
\(\mu_{\rm crit}\), replace \(4\) by \(d+1\). The classical
fidelity level becomes \(2/(d+1)\), and the baseline condition is
\(1/(d+1)<v_{\rm enc}\leq v_{\rm u}\leq1\). The crossover
is unchanged. The cases \(\mu=1\), identical curves and thresholds
outside \([0,1]\) follow the same rules. A baseline at or below
the classical bound never exceeds it under these curves.
For an ideal unencoded baseline, the interval exists precisely when
\(v_{\rm enc}>1/(1+d\mu)\). For qutrits, this is also the
entanglement-comparison condition because ideal unencoded
transmission reaches both boundaries at \(\varepsilon=3/4\).

These crossings concern fidelity. To decide when the unencoded
protocol must destroy entanglement, we now use its channel description.

\subsubsection{When unencoded transmission becomes entanglement breaking}
\label{subsec:depolarising-entanglement-breaking}

The boundary channel and the classical bound for retained trials
both follow from the same Haar averaging identity. Writing
\(\rho_\psi=\ketbra{\psi}\), it is
\begin{equation}
 d\int d\psi\,\Tr(A\rho_\psi)\rho_\psi
 =\frac{A+\Tr(A)I_d}{d+1}.
 \label{eq:haar-measure-prepare-identity}
\end{equation}
It follows by contracting the second moment
\(\int\rho_\psi^{\otimes2}d\psi=(I+\mathsf S)/[d(d+1)]\),
where \(\mathsf S\) exchanges the two copies.

\begin{proposition}[Qutrit depolarising boundary]
\label{prop:qutrit-depolarising-entanglement-breaking}
For \(0\leq\varepsilon\leq1\),
\(\mathcal D_\varepsilon^{(3)}=(1-\varepsilon)\id+\varepsilon\Delta_3\)
is entanglement breaking exactly when \(\varepsilon\geq3/4\).
\end{proposition}

\begin{proof}
Below \(3/4\), its average fidelity
\(1-2\varepsilon/3\) exceeds the qutrit measure-and-prepare
bound \(1/2\), so it cannot be entanglement breaking.
At the boundary, Eq.~\eqref{eq:haar-measure-prepare-identity}
with \(d=3\) gives
\begin{equation}
 3\int d\psi\,\langle\psi|A|\psi\rangle\ketbra{\psi}
 =\frac{A+\Tr(A)I_3}{4}
 =\mathcal D_{3/4}^{(3)}(A).
 \label{eq:depolarising-boundary-measure-prepare}
\end{equation}
The operators \(3\ketbra{\psi}\,d\psi\) form a measurement,
and the output is the corresponding state \(\ketbra{\psi}\).
Thus the boundary channel measures and prepares. For
\(3/4\leq\varepsilon\leq1\),
\begin{equation*}
 \mathcal D_\varepsilon^{(3)}
 =4(1-\varepsilon)\mathcal D_{3/4}^{(3)}
 +(4\varepsilon-3)\Delta_3.
\end{equation*}
This is a convex combination of entanglement-breaking channels.
\end{proof}

Under the complete channel model, subsequent baseline errors
cannot restore entanglement. The unencoded protocol is thus
entanglement breaking for \(\varepsilon\geq3/4\), and possibly
earlier. For the entanglement comparison, evaluate the encoded line at
this boundary, rather than at the unencoded classical threshold.
The encoded protocol still exceeds \(1/2\) there precisely when
\begin{equation*}
 v_{\rm enc}(1+3\mu)>1
 \quad\Longleftrightarrow\quad
 F_{\rm enc}^{(0)}>\frac{1+\mu}{1+3\mu}.
\end{equation*}
This yields the interval in
Proposition~\ref{prop:encoded-teleportation-entanglement}.
For \(\mu=1\), it is \([3/4,1]\) whenever
\(F_{\rm enc}^{(0)}>1/2\). A fidelity at or below \(1/2\)
does not by itself imply that a channel is entanglement breaking.

\subsubsection{A numerical comparison of the two intervals}

For Fig.~\ref{fig:experimental-teleportation},
\(v_{\rm u}=19/20\), \(v_{\rm enc}=4/5\), and \(\mu=1/5\).
The curves are
\begin{equation}
\begin{aligned}
 F_{\rm u}&=\frac{29}{30}-\frac{19}{30}\varepsilon,\\
 F_{\rm enc}&=\frac{13}{15}-\frac{32}{75}\varepsilon.
\end{aligned}
 \label{eq:reference-fidelity-curves}
\end{equation}
Table~\ref{tab:reference-teleportation-thresholds} orders the four
relevant noise strengths. On \([14/19,55/64)\), only the encoded
fidelity exceeds \(1/2\). This interval has width
\(149/1216\simeq0.122533\). The entanglement comparison holds
on the smaller interval \([3/4,55/64)\). At its lower endpoint,
\(F_{\rm enc}(3/4)=41/75\), a margin of \(7/150\) above
\(1/2\). For these baselines, an interval where only encoding exceeds the
classical fidelity bound exists exactly when
\(\mu>15/224\simeq0.0669643\). The sufficient condition in
Proposition~\ref{prop:encoded-teleportation-entanglement} for the
entanglement comparison is instead \(\mu>1/12\). These values
are chosen examples, not device data.

\begin{table}[tbp]
\caption{Noise strengths for the example in
Fig.~\ref{fig:experimental-teleportation}. The entanglement-breaking
statement uses the complete channel model.}
\label{tab:reference-teleportation-thresholds}
\begin{ruledtabular}
\begin{tabular}{c >{\raggedright\arraybackslash}p{0.76\columnwidth}}
\(\varepsilon\) & Meaning\\ \hline
\(15/31\) & The two fidelities cross\\
\(14/19\) & Unencoded fidelity reaches \(1/2\)\\
\(3/4\) & Unencoded channel is guaranteed to be entanglement breaking\\
\(55/64\) & Encoded fidelity reaches \(1/2\)
\end{tabular}
\end{ruledtabular}
\end{table}

\subsection{Averaging fidelity after trials are discarded}
\label{subsec:branch-resolved-data}

\subsubsection{How discarding trials changes the noise weights}

A branch contributes in proportion to how often it is sampled and
how often its trials are retained. If branch \(j\) is sampled
with probability \(\pi_j\), retains a fraction \(h_j\), and has
fidelity \(F_j\) among its retained trials, then
\begin{equation*}
 F^{\rm acc}=\frac{\sum_j\pi_jh_jF_j}{\sum_j\pi_jh_j}.
\end{equation*}
For example, equal sampling of two branches with retention
probabilities \(1\) and \(1/2\) gives retained weights \(2/3\)
and \(1/3\). Inputs are reweighted within each branch in the
same way. Thus \(F_j\) averages the actual retained data, rather
than giving every input equal weight after normalisation.

To include branches that retain no trials, use unnormalised
outputs. Label the branches \(j\in\{{\rm id},{\rm c},{\rm i}\}\),
with probabilities \(\pi_j\) from
Table~\ref{tab:weyl-branch-actions}, chosen independently of the
input. For protocol \(\alpha\), let \(\mathcal S_\alpha^j\)
be the completely positive, trace-nonincreasing map for outputs
retained in branch \(j\), averaged over the original uniform
sampling of its labels.
For \(\rho_\psi=\ketbra{\psi}\), define
\begin{equation*}
 h_\alpha^j=\int d\psi\,\Tr[\mathcal S_\alpha^j(\rho_\psi)],
 \qquad
 q_\alpha^j=\int d\psi\,
 \langle\psi|\mathcal S_\alpha^j(\rho_\psi)|\psi\rangle.
\end{equation*}
Here \(d\psi\) is normalised Haar measure on pure logical inputs.
The first quantity is the probability of retaining a trial in that
branch, averaged over inputs. The second averages the overlap with
the unnormalised output. The probability that a trial is retained and the
average fidelity over all retained trials are
\begin{equation}
\begin{aligned}
 P_\alpha^{\rm acc}&=\sum_j\pi_jh_\alpha^j,\\
 F_\alpha^{\rm acc}
 &=\frac{\sum_j\pi_jq_\alpha^j}{P_\alpha^{\rm acc}},
 \qquad P_\alpha^{\rm acc}>0.
\end{aligned}
 \label{eq:branch-resolved-observables}
\end{equation}
Positivity gives \(0\leq q_\alpha^j\leq h_\alpha^j\).
If \(h_\alpha^j=0\), the branch contributes zero without requiring
a definition of its fidelity. Otherwise
\(F_\alpha^j=q_\alpha^j/h_\alpha^j\), recovering the weighted
average above.

The main-text curves follow if all three branches have the same
positive probability of retaining a trial within each protocol, and
their fidelities have the values in
Table~\ref{tab:accepted-branch-fidelities}. These scalar conditions
do not establish the factorisation into channels used in the
entanglement comparison.

\begin{table}[tbp]
\caption{Branch fidelities sufficient to recover the main-text
curves when trials from every branch are equally likely to be retained.
The qutrit protocol has \(d=3\).}
\label{tab:accepted-branch-fidelities}
\begin{ruledtabular}
\begin{tabular}{l c c}
Branch & Unencoded & Encoded\\ \hline
No randomisation & \(F_{\rm u}^{(0)}\) & \(F_{\rm enc}^{(0)}\)\\
Shared label & \(1/d\) & \(F_{\rm enc}^{(0)}\)\\
Independent labels & \(1/d\) & \(1/d\)
\end{tabular}
\end{ruledtabular}
\end{table}

\subsubsection{The classical fidelity bound for retained trials}

Discarding trials does not raise the classical fidelity bound when one
copy of a uniformly sampled pure input is provided per trial and
fidelity is averaged over all retained trials. This remains true
even when the probability of retaining a trial depends on the input.
For a measure-and-prepare protocol, the map for retained outputs is
\(\mathcal S_{\rm cl}(\rho)=\sum_y\Tr(M_y\rho)\sigma_y\),
where \(M_y\geq0\), \(\sum_yM_y\leq I_d\), and each
\(\sigma_y\) is a density operator. The first Haar moment gives
a retained probability \(\sum_y\Tr M_y/d\). Set \(A=M_y\) in
Eq.~\eqref{eq:haar-measure-prepare-identity} and take the overlap
with \(\sigma_y\). Each numerator contribution is
\([\Tr M_y+\Tr(M_y\sigma_y)]/[d(d+1)]\). Provided the
retained probability is nonzero,
\begin{equation}
 F_{\rm MP}^{\rm acc}
 =\frac1{d+1}\left(1+
 \frac{\sum_y\Tr(M_y\sigma_y)}{\sum_y\Tr M_y}\right)
 \leq\frac{2}{d+1}.
 \label{eq:haar-heralded-classical-bound}
\end{equation}
Every prepared state obeys \(0\leq\sigma_y\leq I_d\), so the
ratio is at most one, regardless of which outcomes are discarded.
A rank-one
projective measurement followed by preparation of the detected
basis state attains it. In particular, the qutrit bound is \(1/2\).

If a trial is retained with the same probability \(h>0\) for every
input, then \(\mathcal S/h\) is a CPTP channel. Input-dependent
normalisation generally gives a nonlinear map, although the average
fidelity over retained trials remains well defined. For a
measure-and-prepare map it obeys the bound above. A different input
ensemble, or averaging separately normalised input fidelities,
requires its own derivation of the classical bound.

\subsection{Implementing and testing the noise model}
\label{subsec:experimental-calibration-protocol}

\subsubsection{Sampling errors with four or five qutrits}

First measure both baseline fidelities and the fractions of trials
retained at \(\varepsilon=0\). Applying a unitary 2-design to one
fixed pure state supplies a finite input ensemble for the Haar
averages~\cite{DankertEtAl2009}.
Then sample the branches and Weyl labels with the probabilities in
Eq.~\eqref{eq:weyl-noise-channel}. Apply \(D_g\) to a \(V\) factor
and \(\overline D_g\) to a \(V^*\) factor. Pad the identity branch
to the same duration and nominal pulse depth.

In the five-qutrit realisation, an isometry
\(J_*:V^*\to\wedge^2V\) represents the conjugate factor by an
antisymmetric pair. Since \(D_g\in SU(3)\),
\begin{equation}
 (D_g\otimes D_g)J_*=J_*\overline D_g.
 \label{eq:weyl-conjugate-pair-implementation}
\end{equation}
To express \(K\) in this realisation, choose
\begin{equation}
 J_*\ket{\bar a}
 =\frac{1}{\sqrt2}\sum_{b,c=0}^{2}
 \epsilon_{abc}\ket b\otimes\ket c,
 \qquad a\in\{0,1,2\},
 \label{eq:dual-isometry-explicit}
\end{equation}
with \(\epsilon_{abc}\) the Levi-Civita tensor and \(\epsilon_{012}=1\).
This is an isometry onto the antisymmetric subspace
\(\wedge^2V\subset V\otimes V\). Applying it to the conjugate factor
of the normalised singlet
\(\ket\Omega=3^{-1/2}\sum_{a=0}^{2}\ket a\otimes\ket{\bar a}\),
we find
\begin{equation}
\begin{aligned}
 (I_V\otimes J_*)\ket\Omega
 &=\frac{1}{\sqrt6}\sum_{a,b,c=0}^{2}
 \epsilon_{abc}\ket a\otimes\ket b\otimes\ket c\\
 &=:\ket{A_3}.
\end{aligned}
 \label{eq:encoded-antisymmetric-singlet}
\end{equation}
The vector \(\ket{A_3}\) is the normalised totally antisymmetric
three-qutrit state. The interaction \(K=3\ketbra\Omega\) therefore becomes
\begin{equation}
\begin{aligned}
 (I_V\otimes J_*)K(I_V\otimes J_*^\dagger)
 &=3\ketbra{A_3}\\
 &=3P_{\wedge^3V}.
\end{aligned}
 \label{eq:encoded-singlet-interaction}
\end{equation}
Here we view \(J_*\) as an isometry into \(V\otimes V\).
This operator vanishes outside \(V\otimes\wedge^2V\) and acts
on one fundamental qutrit together with the two representing the conjugate
factor. It is the three-qutrit interaction described in
Sec.~\ref{sec:heralded-cost}.

The independent branch still samples four labels: three for the
individual qutrits and one shared by the two members of the pair.
The collective branch applies one label to all five qutrits.
Choosing five independent labels would change the noise model and
need not preserve the antisymmetric subspace.

In either realisation, independently chosen labels can move the
state outside the selected noiseless sector. Recovery retains these
trials and replaces that component by \(I_3/3\). Trials with a
signalled loss or another detected failure are discarded and counted
separately. The deliberately applied Weyl errors
remain part of the noise: they are not undone or absorbed into a
change of the receiver's reference frame.

\subsubsection{Checking which trials are retained}

For each setting of \(\varepsilon,\mu\), record how many trials
are attempted and retained for each input, noise branch and sampled
label, together with their fidelity data. This allows one to check
whether discarding trials favours particular inputs or errors.
Use the attempted and retained counts together with the fidelity data
to calculate both quantities in Eq.~\eqref{eq:branch-resolved-observables},
averaging with the original input sampling weights. The retained inputs
need not themselves form a 2-design. The resulting averages test
the fidelity conditions for each branch in
Table~\ref{tab:accepted-branch-fidelities} and the predicted curves.

To test the sampling assumption in
Sec.~\ref{sec:baseline-calibration}, compare the retained
probabilities jointly for each input, branch and label at fixed
\(\varepsilon,\mu\) within each protocol. Separate averages
can hide selection bias. For example, choose two branches equally and
retain their outputs using the qutrit measurement effects
\begin{equation*}
 M_\pm=\tfrac12[I_3\pm t\,\operatorname{diag}(1,-1,0)],
 \qquad 0<t\leq1.
\end{equation*}
The overall retained probability is \(1/2\) for every input,
and the Haar-averaged probability is \(1/2\) in each branch.
Yet within each branch, this probability depends on the input. One branch favours
level \(\ket0\) and the other favours level \(\ket1\), so their
biases cancel in the separate averages. The retained branch weights
still depend on the input.

\subsubsection{Testing the complete channel model}

The entanglement comparison additionally requires a model of the
complete protocol that places a fixed baseline channel after the
ideal effective noise channel. Even if errors are independent of the
branch and label, this order of composition needs a separate justification.
Process tomography can test the proposed channel description.
Baseline fidelity, isotropy tests and randomised benchmarking alone
do not establish it. Randomised and cycle benchmarking characterise
averaged gate errors and stochastic Weyl components
\cite{MagesanEtAl2011,MorvanEtAl2021,GossEtAl2022,GossEtAl2024}. For temporally correlated quantum noise, protected-gate
performance can depend on control history and bath
evolution~\cite{BurgelmanEtAl2025}. A fixed baseline must therefore be
checked over the protocol sequences used in the comparison.

A simple example shows why even the whole fidelity curve is
insufficient to identify the channel. The qutrit reset channel
\begin{equation}
 \mathcal R_\varepsilon(\rho)
 =(1-\varepsilon)\rho
 +\varepsilon\Tr(\rho)\ketbra{0}
 \label{eq:reset-channel-fidelity-example}
\end{equation}
has average fidelity \(1-2\varepsilon/3\), exactly the
ideal depolarising curve: the fixed reset output contributes
an average overlap \(1/3\). However, it is not entanglement
breaking for any \(\varepsilon<1\). Store one half of a maximally
entangled qubit pair in levels \(\ket{1},\ket{2}\). Projecting the output onto
those levels recovers the entangled state with probability
\(1-\varepsilon>0\), since the reset output \(\ket{0}\) is
orthogonal to them. A local projection cannot create entanglement
from a separable state. This proves that the output before the
projection remains entangled, even where the average fidelity is
at or below \(1/2\).

The baseline refers to the complete teleportation protocol, rather
than directly to noise on its resource. For example, a resource
error \(\mathcal J\) after depolarising transmission gives
\(\mathsf T_d(\mathcal J\circ\mathcal D_\xi^{(d)})
=\mathsf T_d(\mathcal J)\circ\mathcal D_\xi^{(d)}\).
The corresponding baseline is \(\mathsf T_d(\mathcal J)\).
A Weyl channel twirl generally produces a stochastic Weyl channel,
not a depolarising channel. A Haar or suitable unitary 2-design
conjugation twirl makes a channel isotropic while preserving average
fidelity~\cite{Nielsen2002,DankertEtAl2009,MorvanEtAl2021}.

Errors in the applied gates can be included in the baseline only
when this complete channel model remains valid. If errors, loss,
leakage from the qutrit levels or detection depend on branch or
label, use the measured branch maps and
Eq.~\eqref{eq:branch-resolved-observables}. An entanglement comparison
then needs a channel description appropriate to those effects.

\subsection{Physical resources and transmission costs}
\label{app:resource-accounting}

\subsubsection{Transmission loss and the cost of obtaining an output}

Assume that each transmitted physical qutrit arrives independently
with probability \(0<\eta\leq1\). Assume also that every loss is detected
without false signals and that decoding requires all transmitted
qutrits, with all of them sent on each attempt. To keep the fidelity
on retained trials unchanged, loss
must be independent of the input, noise branch,
Weyl labels and processing errors. If the protocol without loss is
described by a CPTP channel \(\Lambda_n\), independent loss in
transmission of its \(n\) qutrits gives the erasure channel
\begin{equation}
 \mathcal T_n(\rho)
 =\eta^n\Lambda_n(\rho)
 +(1-\eta^n)\Tr(\rho)\ketbra e,
 \label{eq:heralded-erasure-channels}
\end{equation}
where \(\ket e\) is orthogonal to the logical output space and
\(\Lambda_n\) is the channel on trials with no erasure.
A full attempt succeeds with probability \(\eta^n\). With
independent repeated attempts, the mean number of attempts until an
output arrives, including the successful attempt, is \(1/\eta^n\).
Each attempt sends all \(n\)
qutrits, giving the mean cost \(n/\eta^n\) in
Table~\ref{tab:loss-transmission-costs}. Equivalently, the expected
number of outputs per physical transmission is \(\eta^n/n\).
Under these assumptions loss reduces this rate while leaving the
fidelity of received outputs unchanged.

\begin{table}[tbp]
\caption{Costs from transmission loss when all \(n\) qutrits must arrive.
The unencoded protocol has \(n=1\), and the two encoded
realisations have \(n=4\) and \(n=5\).}
\label{tab:loss-transmission-costs}
\begin{ruledtabular}
\begin{tabular}{>{\raggedright\arraybackslash}p{0.69\columnwidth}c}
Quantity & Value\\ \hline
Arrival probability per logical attempt & \(\eta^n\)\\
Expected physical transmissions per received output & \(n/\eta^n\)\\
Encoded/unencoded arrival ratio per logical attempt & \(\eta^{n-1}\)\\
Encoded/unencoded ratio of outputs per physical transmission & \(\eta^{n-1}/n\)
\end{tabular}
\end{ruledtabular}
\end{table}

For different transmittances of the two types, the encoded arrival
probability is \(\eta_V^r\eta_{V^*}^s\). Representing every
\(V^*\) factor by two ordinary qutrits gives \(\eta_V^{r+2s}\).
Preparation, Bell-measurement and detection efficiencies must also
be included. When these are independent stages, their probabilities
multiply the arrival probability. Otherwise measure the overall
probability of obtaining a retained output. Reducing the cost of
requiring every qutrit to arrive calls for a loss-tolerant encoding,
multiplexing or another method for increasing the transmission rate~\cite{IkutaEtAl2011}.

\subsubsection{Entanglement and communication resources}

Only Bob's share of the entangled resource is encoded. Before
transmission noise is applied, the resource state is, up to a local
isometry and tensor ordering,
\begin{equation}
 \sigma_{A:UW}=\Phi_{AW}\otimes\tau_U,
 \label{eq:encoded-resource-gauge-product}
\end{equation}
where \(\Phi_{AW}\) is a maximally entangled state of local
dimension \(d\). Attaching and discarding \(\tau_U\) are local operations
on Bob's side. Applying monotonicity in both directions shows that
distillable entanglement and entanglement of formation both remain
\(\log_2d\), for a pure or mixed gauge state. The Bell measurement
still has \(d^2\) outcomes, which can be communicated using two
classical dits.
This counts the entanglement supplied to the protocol. Noise during
distribution can reduce the entanglement that reaches the receiver.

Table~\ref{tab:resource-ledger-master} gives the qutrit comparison.
The selected block has dimension \(6\times3=18\) inside the
81-dimensional physical space. These dimensions are not gate
counts. Efficient Schur transforms are known for both pure and
mixed tensor spaces~\cite{BaconChuangHarrow2006,Nguyen2023MixedSchur}.
Those constructions use general circuit gates. This work does not
supply an encoding or recovery circuit for the chosen physical
interactions, or a gate count for either operation.
For each realisation, both the baseline fidelity and the probability
of retaining a trial must be measured. A fidelity advantage alone
does not establish a communication-rate advantage at fixed resources.

\begin{table*}[t]
\caption{Resources when only Bob's share is encoded and transmitted.
Entanglement is counted before transmission noise. Arrival
probabilities use the independent loss assumptions above.}
\label{tab:resource-ledger-master}
\begin{ruledtabular}
\begin{tabular}{>{\raggedright\arraybackslash}p{0.20\textwidth}
>{\raggedright\arraybackslash}p{0.23\textwidth}
>{\raggedright\arraybackslash}p{0.46\textwidth}}
Resource & Unencoded protocol & Encoded protocol\\ \hline
Transmitted qutrits & one & four with both types directly available, or five using an antisymmetric pair\\
Distillable entanglement and entanglement of formation & \(\log_2 3\) & \(\log_2 3\)\\
Bell outcomes and messages & nine outcomes, two trits & nine outcomes, two trits, with logical correction before recovery\\
Baseline fidelity & \(F_{\rm u}^{(0)}\) & \(F_{\rm enc}^{(0)}\)\\
Arrival probability & \(\eta\) & \(\eta^4\) or \(\eta^5\)\\
Control interactions & ordinary teleportation controls & exchanges and \(K=3P_\Omega\). In the five-qutrit realisation, \(K\) acts on three physical qutrits\\
Additional preparation and processing & none from an encoding & gauge preparation, encoding, logical correction and recovery. Their synthesis costs are not supplied
\end{tabular}
\end{ruledtabular}
\end{table*}

\subsubsection{Alternative logical qutrits}
\label{app:compact-blocks}

The decomposition in Eq.~\eqref{eq:H31-logical-qutrit-decomposition}
also contains \(W_{3,1}^{0,1}\), a second logical qutrit.
Its block has dimension \(3\times3=9\), compared with 18 for the
block containing \(W_{3,1}^{2,0}\), but both choices use the same four physical
factors. Proposition~\ref{prop:analytic-31-semiuniversality}
proves independent \(SU(3)\) control on these two spaces.
We select \(W_{3,1}^{2,0}\) because it is the target of the pulse
sequence in Sec.~\ref{subsec:explicit-sector-control}.

Four ordinary qutrits also contain noiseless qutrits:
\(V^{\otimes4}\cong U_{4,0}\oplus3U_{2,1}
\oplus2U_{0,2}\oplus3U_{1,0}\).
This decomposition alone does not imply the independent control
proved here for systems containing both types.

\subsection{Other noise models and bounds on fidelity loss}
\label{subsec:collective-extensions}

The Weyl mixture isolates the effects of shared and independent
labels. It is not a microscopic model of arbitrary relaxation,
local dephasing, leakage, unequal rotations, spatial or temporal
correlations, or crosstalk. Some other collective processes still
preserve the logical state exactly. Additional logical noise
requires a separate bound or a measured channel model.

\subsubsection{Other collective processes that preserve the logical state}

Lemma~\ref{lem:collective-factorization-teleportation} applies to
any channel whose Kraus operators act only on the gauge factor.
This includes channels obtained by averaging collective rotations
\(R_{r,s}(U)\) over any distribution \(\nu\). For an unencoded qutrit, set
\(M_\nu=\int d\nu(U)|\Tr U|^2\). Applying a sampled rotation
with probability \(\varepsilon\), and otherwise doing nothing,
gives average fidelity \(1-\varepsilon(9-M_\nu)/12\), by
\(F_{\rm avg}(U)=(|\Tr U|^2+3)/12\)~\cite{Nielsen2002}.
This scalar value does not determine composition with an arbitrary
baseline channel. A Haar or unitary 2-design conjugation twirl
makes the channel isotropic, with rescaled fidelity
\(1-\varepsilon(9-M_\nu)/8\). After the baseline, the rescaled
fidelities are then, with the same assumed order of channels,
\(v_{\rm u}[1-\varepsilon(9-M_\nu)/8]\) and \(v_{\rm enc}\)
\cite{DankertEtAl2009}. Without that additional twirl, the composed
channels should be measured directly. This argument applies to
other specified random unitary ensembles without interpreting
\(\mu\) as a microscopic noise parameter.

Continuous collective dynamics provide another example. For
anti-Hermitian \(T_a\in\mathfrak{su}(3)\),
\begin{equation}
 J_a=dR_{r,s}(T_a)
 =\sum_{i=1}^r T_a^{(i)}+
 \sum_{j=1}^s\overline{T_a}^{(\bar j)}.
 \label{eq:matched-markov-generator}
\end{equation}
Each \(J_a\) acts on the gauge factor alone in every sector.
If the Hamiltonian and jump operators of a Lindblad generator
preserve each sector and act only on its gauge factor, the
resulting channels satisfy the lemma
\cite{TicozziViola2008}. Engineered noncommuting collective dynamics
have been demonstrated experimentally~\cite{ViolaEtAl2001}.

\subsubsection{Comparing measured decay rates}
\label{app:markovian}

A device may instead be described by fitted exponential decays of
the rescaled fidelities,
\begin{equation}
 a_{\rm u}(t)=v_{\rm u}e^{-\Gamma_{\rm u}t},
 \qquad a_{\rm enc}(t)=v_{\rm enc}e^{-\Gamma_{\rm L}t}.
 \label{eq:calibrated-markov-decays}
\end{equation}
The rescaling is the same as in
Eq.~\eqref{eq:rescaled-fidelity-coordinate}, and
\(\Gamma_{\rm L}\) describes residual logical noise. It vanishes
in the exact collective model. Assume
\(0\leq\Gamma_{\rm L}<\Gamma_{\rm u}\),
\(v_{\rm enc}\leq v_{\rm u}\), and that both initial fidelities exceed
\(2/(d+1)\). Solving for the crossing of the curves and for
equality with the classical bound gives
\begin{equation*}
\begin{aligned}
 t_\times&=\frac{\ln(v_{\rm u}/v_{\rm enc})}
 {\Gamma_{\rm u}-\Gamma_{\rm L}},\\
 t_{\rm u}^{\rm MP}&=\frac{\ln[(d+1)v_{\rm u}]}{\Gamma_{\rm u}},
 \qquad
 t_{\rm enc}^{\rm MP}=\frac{\ln[(d+1)v_{\rm enc}]}{\Gamma_{\rm L}}.
\end{aligned}
\end{equation*}
The last expression assumes \(\Gamma_{\rm L}>0\). At
\(\Gamma_{\rm L}=0\), the encoded threshold is infinite.
Only the encoded protocol exceeds the bound on
\([t_{\rm u}^{\rm MP},t_{\rm enc}^{\rm MP})\) when these endpoints
satisfy \(t_{\rm u}^{\rm MP}<t_{\rm enc}^{\rm MP}\).
These are fidelity comparisons. The fitted rates are
not identified with the parameters of the finite Weyl model.

\subsubsection{A fidelity bound from distance to ideal transmission}
\label{app:weak-breaking}

Let \(\mathcal S_\zeta\) be the completely positive,
trace-nonincreasing map for the transmission outputs retained in the
presence of additional noise of strength \(\zeta\). Assume that
\(h_\zeta=\Tr[\mathcal S_\zeta(\operatorname{Enc}(\rho))]>0\)
is independent of the logical input \(\rho\). The normalised
logical channel and its distance from ideal transmission are
\begin{equation*}
 \Lambda_\zeta=\frac{\operatorname{Rec}\circ\mathcal S_\zeta
 \circ\operatorname{Enc}}{h_\zeta},
 \qquad
 \varepsilon_\diamond=\frac12
 \|\Lambda_\zeta-\id_d\|_\diamond.
\end{equation*}
Input-independent normalisation makes \(\Lambda_\zeta\) CPTP.
The distance \(\varepsilon_\diamond\) measures the total deviation
of the recovered transmission from the identity channel. The reference
fidelity \(F_{\rm enc}^{(0)}\) in the bounds below is the calibrated
baseline fidelity of Sec.~\ref{sec:baseline-calibration}.
Assume, as before, that a fixed CPTP baseline channel
\(\mathcal E_{\rm enc}^{(0)}\) follows the effective teleportation
channel \(\mathsf T_d(\Lambda_\zeta)\). The Weyl average in
teleportation and subsequent composition with a channel cannot
increase diamond distance. Hence
\begin{equation*}
 \tfrac12\|\mathcal E_{\rm enc}^{(0)}\circ
 \mathsf T_d(\Lambda_\zeta)-\mathcal E_{\rm enc}^{(0)}\|_\diamond
 \leq\varepsilon_\diamond.
\end{equation*}
For each pure input, the overlaps of the two output states with
that input differ by at most \(\varepsilon_\diamond\). Averaging gives
\begin{equation}
 F_{\rm enc}^{(\zeta)}\geq
 F_{\rm enc}^{(0)}-\varepsilon_\diamond.
 \label{eq:weak-breaking-general-baseline}
\end{equation}
This bound needs no isotropy assumption on the baseline.

A depolarising baseline gives a sharper bound. Suppose
\(\mathcal E_{\rm enc}^{(0)}(A)
=v_{\rm enc}A+(1-v_{\rm enc})\Tr(A)I_d/d\), with
\(0\leq v_{\rm enc}\leq1\). The same distance argument without
baseline errors gives
\(F_{\rm avg}(\Lambda_\zeta)\geq1-\varepsilon_\diamond\).
Haar invariance implies
\(F_{\rm avg}(\mathsf T_d(\Lambda_\zeta))
=F_{\rm avg}(\Lambda_\zeta)\), so
\begin{equation}
\begin{aligned}
 F_{\rm enc}^{(\zeta)}
 &=v_{\rm enc}F_{\rm avg}(\Lambda_\zeta)
   +\frac{1-v_{\rm enc}}d\\
 &\geq F_{\rm enc}^{(0)}-v_{\rm enc}\varepsilon_\diamond.
\end{aligned}
 \label{eq:weak-breaking-fidelity}
\end{equation}
The nonnegative \(v_{\rm enc}\) assumption holds whenever the
baseline exceeds the classical fidelity bound. Comparing the
applicable lower bound with the observed unencoded fidelity \(F_{\rm u}^{\rm obs}\)
or with \(2/(d+1)\)
gives sufficient conditions for a fidelity advantage or for
exceeding the classical fidelity bound, respectively. Both protocols
must be evaluated at the same physical noise parameters.
If the unencoded protocol follows the finite Weyl curve, substitute
its earlier expression for \(F_{\rm u}^{\rm obs}\).
These conditions are sufficient, not necessary. The size of
\(\varepsilon_\diamond\) depends on how the perturbation acts
after encoding and recovery. No universal quadratic scaling with
the additional noise strength is assumed~\cite{BaconLidarWhaley1999}.

% \vspace{0.2cm}

% \section*{Data and code availability}

% No experimental or third-party datasets were used. The data generated
% here include exact Lie-algebra dimension checks
% (Table~\ref{tab:small-system-ranks}) and optimised pulse parameters
% (Table~\ref{tab:explicit-pulse-areas-31}).
% The exact calculations are described in
% Appendix~\ref{subsec:exact-rank-procedure}.

% For Sec.~\ref{subsec:explicit-sector-control} and
% Appendix~\ref{subsec:explicit-pulses-31}, the self-contained package
% \nolinkurl{section_III_C_reproducibility} contains Python code for
% constructing the exact interaction matrices and logical bases, searching
% for pulse sequences, and independently verifying them. It includes all
% 24 parameters of Table~\ref{tab:explicit-pulse-areas-31} in machine-readable
% form, the optimisation records, and saved verification results on the
% 81-dimensional physical space and at 90-digit precision on the logical
% blocks, without phase adjustment. Installation instructions, a description
% of the model and methods, and regression tests are included.
% The figures use the analytical expressions and parameter values given
% in the paper. Public deposition of the supporting code and data is
% pending. The repository link or DOI remains to be added.

% AUTHOR TODO: Deposit the supporting code and data and replace the pending
% sentence with the actual repository/DOI citation before submission.
% The section_III_C_reproducibility package covers Section III C and Table VI;
% its exact rank checks concern (3,1), not all registers in Table V.
% Include the separate rank-check and figure-generation packages in the release.

\let\originalbibfont\bibfont
\renewcommand{\bibfont}{\originalbibfont\setlength{\baselineskip}{10.3pt}}
\bibliography{references}

\end{document}